\documentclass[reprint,superscriptaddress,aps,pra,floatfix]{revtex4-2}
\usepackage{graphicx}
\graphicspath{{./}{./figs/}}
\usepackage{dcolumn}
\usepackage{bm}
\usepackage{soul}
\usepackage{mathdots}
\usepackage{booktabs} 
\usepackage{array}    
\usepackage{amsmath,amssymb}  
\usepackage[colorlinks,citecolor=blue]{hyperref}

\usepackage{enumitem}
\setlist{nosep,leftmargin=*,afterlabel=~}
\usepackage{amsthm}
\newtheorem{theorem}{Theorem}

\newtheorem{statement}[theorem]{Statement}

\newtheorem{lemma}[theorem]{Lemma}

\theoremstyle{definition}

\newtheorem{example}{Example}
\newtheorem{algorithm}{Algorithm}

\def\bs#1{\boldsymbol{#1}}

\def\diag{\mathop{\rm diag}}
\def\ord{\mathop{\rm ord}}

\def\rank{\mathop{\rm rank}}
\def\MM{\mathop{\mathrm M}\nolimits}

\def\css{\mathop{\rm CSS}}
\def\lcm{\mathop{\rm lcm}}
\def\wgt{\mathop{\rm wgt}\nolimits}
\def\res{\mathop{\rm Res}\nolimits}
\def\im{\mathop{\rm im}}

\usepackage{tabularray}
\usepackage{xcolor}

\begin{document}

\title{Algebra of Bivariate-Bicycle Surface Codes}

\author{Renyu Wang}
\affiliation{Department of Physics \& Astronomy, University of
  California, Riverside, California 92521 USA}

\author{Leonid P. Pryadko}%
\email{pryadko@google.com}
\affiliation{Google Quantum AI,  Santa Barbara, California 93117, USA}
\affiliation{Department of Physics \& Astronomy, University of
  California, Riverside, California 92521 USA}

\date{\today}

\begin{abstract}
  We relate the properties of bivariate-bicycle-surface (BBS) codes,
  constructed from a pair of bivariate polynomials over a finite
  field, to the number and location of their common roots in the
  extension field. The number of roots $(x,y)$ with finite, non-zero
  coordinates---counted with algebraic multiplicity---determines the
  dimension of the codes. This dimension is invariant under monomial
  automorphisms of the Laurent polynomial ring. Conversely, roots with
  zero or infinite $x$- or $y$-coordinates indicate that specialized 
  generators are required near the corresponding boundary (e.g., the
  left or right boundary for a root where $x$ is zero or infinite,
  respectively). These roots can appear or disappear under monomial
  transformations, which reveals the structure of tilted
  boundaries. Based on these results, we formulate a prescription for
  constructing BBS codes that works for regions with rectangular,
  diagonal, and arbitrarily tilted boundaries. A key advantage of this
  approach is that no corner corrections are needed, provided the
  polynomials satisfy orientation-specific edge conditions.
\end{abstract}

\maketitle

\section{Introduction}
\label{sec:intro}

Topological surface codes
\cite{Bravyi-Kitaev-1998,Dennis-Kitaev-Landahl-Preskill-2002} are at
the forefront of current efforts to achieve scalable quantum
computation
\cite{google-2023-suppressing,Paetznik-etal-trapped-ions-2024,%
  Bluvstein-etal-Lukin-2024,google_Quantum_AI-Y1-2025}.  Their primary
advantage is locality, requiring only gates between neighboring
qubits.  Furthermore, surface codes offer relatively high thresholds,
flexibility in both hardware implementation and decoding, and the
ability to perform encoded Clifford
gates\cite{Fowler-Mariantoni-Martinis-Cleland-2012,%
  Terhal-RMP-2015,Roffe-review-2019}.  The main disadvantage of
surface codes is their poor encoding rates, a fundamental limitation
imposed by code locality in two
dimensions\cite{Bravyi-Terhal-2009,Bravyi-Poulin-Terhal-2010}.

This encoding rate limitation is resolved in the more general
family of quantum low-density parity-check (LDPC) codes.  Bounded
stabilizer generator weights ensure fault-tolerance, provided the code
distance grows as a logarithm of the block length or
faster\cite{Kovalev-Pryadko-FT-2013,Gottesman-overhead-2014,%
  Dumer-Kovalev-Pryadko-bnd-2015}.  Quantum LDPC codes include
``good'' code families with finite encoding rates and finite relative
distances\cite{Panteleev-Kalachev-2021,Leverrier-Zemor-2022,%
  Dinur-Hsieh-Lin-Vidick-2023}.  While not as optimal asymptotically,
short quantum LDPC codes with competitive parameters can be obtained
via a two-block ansatz\cite{Kovalev-Pryadko-Hyperbicycle-2013}; this
includes generalized-bicycle (GB)
codes\cite{Panteleev-Kalachev-2019,Wang-Pryadko-2022} based on
circulant matrices, and more general abelian and non-abelian two-block
group-algebra codes\cite{Lin-Pryadko-2023}. In particular,
bivariate-bicycle (BB) codes---a subset of the latter family that
includes the IBM ``gross'' and ``double-gross'' codes---have recently
gained prominence due to their high rates and distances (which
substantially outperform surface codes) and their excellent circuit
performance\cite{Bravyi-etal-Yoder-2023,%
  Eberhardt-Steffan-2024,Symons-Rajput-Browne-2025,%
  Postema-Kokkelmans-2025,Liang-Liu-Song-Chen-2025}.

Conversely, the non-local stabilizer generators make hardware
implementation of quantum LDPC codes difficult.  For the original
family of BB codes\cite{Bravyi-etal-Yoder-2023}
constructed from a pair of bivariate polynomials of weight 3, each
qubit is addressed by six stabilizer generators of constant weight
$w=6$.  With a measurement scheme using one ancillary qubit per
stabilizer generator, the qubit connectivity graph can be represented
as a torus with a nearest-neighbor square-lattice tiling, augmented by
two additional ``long-range'' edges per qubit.  While such graphs are
highly symmetric, with a planar qubit layout the symmetry is
necessarily broken by the edges spanning the wrapped boundaries.  The
need for such edges not only complicates the hardware implementation,
but also rigidly constrains the code size on a given chip.

For conventional toric codes, this issue was resolved by introducing
surface codes with rough and smooth
boundaries\cite{Bravyi-Kitaev-1998}. Not only does this remove the
need for long-range boundary connections, but it also substantially
increases flexibility, as the same fixed square-lattice qubit layout
can be used to implement a wide variety of codes
\cite{Delfosse-Iyer-Poulin-2016}.  In the case of BB codes, codes with
open boundaries have also been
investigated\cite{Eberhardt-Pereira-Steffan-2024,%
  Steffan-etal-Eberhardt-2025,Liang-Eberhardt-Chen-2025,%
  Breuckmann-etal-Steffan-2025}.  However, an analytical understanding
of BB codes with boundaries and their resulting parameters is 
currently available only when they correspond to hypergraph-product (HP)
codes\cite{Tillich-Zemor-2014}.  While the graphical tile construction
\cite{Steffan-etal-Eberhardt-2025,Breuckmann-etal-Steffan-2025} is
straightforward to use, it only applies to a very limited set of
polynomials.

The goal of this work is to relate the algebraic properties of the
defining polynomials to the structure of the corresponding codes with
boundaries.  We provide a first-principles analysis of the structure
of non-trivial chains and co-chains (``codewords'') on finite-width
strips in parent BB complexes on the infinite plane, alongside the
corresponding solutions next to a smooth or rough boundary constructed
by truncating $X$ or $Z$ generators, respectively.  When the
polynomials are mutually prime (the topological order condition on the
infinite plane), we formulate a condition for the boundaries to
support no bounded-weight codewords, as well as a simple prescription
for modifying edge generators when this condition fails.  For
horizontal boundaries ($y=\text{constant}$), this condition requires
that the two polynomials share no common roots $(x,y)$ where $y$ is
zero or infinite.  Such roots may appear or disappear under invertible
monomial transformations of the polynomials, which reveals the
underlying structure of tilted boundaries.  In contrast, the total
number of roots with both $x$ and $y$ finite and non-zero, counted
with algebraic multiplicity, is invariant under such transformations;
it gives both the topological order (TO) total quantum dimension and the
dimension of the BBS code families that can be constructed from a
given pair of polynomials.

The paper is organized as follows. In Section \ref{sec:bg} we list
some background facts and introduce necessary notations.  In
Section~\ref{sec:covering-plane}, we analyze the BB complex on the
infinite plane, derive the exactness condition (equivalent to
ground-state topological order in related Hamiltonians), and consider
the structure of certain infinite-weight solutions of the associated
equations.  Section~\ref{sec:with-boundaries} transitions to finite
geometries, detailing the structure of boundaries and their associated
codewords. In Section~\ref{sec:construction}, we introduce our
algorithm for constructing BBS codes, benchmarking it against existing
methods---such as tile codes
\cite{Steffan-etal-Eberhardt-2025,Breuckmann-etal-Steffan-2025} and
open-boundary codes \cite{Liang-Eberhardt-Chen-2025}---and evaluating
the parameters of the resulting codes. We conclude in
Section~\ref{sec:conclusion} by summarizing the trade-offs inherent to
BBS codes.  Formal proofs and supplementary derivations are deferred
to the Appendices.

\section{Notations}
\label{sec:bg}

\subsection{Classical codes, puncturing, and shortening}
\label{sec:bg_classical}

Let $F\equiv \mathbb{F}_q$ denote a finite field of order $q=p^m$ and
prime characteristic $p$. An $F$-linear code ${\cal C}$ with
parameters $[n,k,d]$ is a linear space of dimension $k$ formed by
$n$-component vectors
$\boldsymbol{c}\equiv (c_0,c_1,\ldots,c_{n-1})\in F^n$. We say that a
code ${\cal C}\equiv{\cal C}_G$ is generated by a matrix $G$ if its
linearly independent rows form a basis of the code, and we denote its
parity check matrix as $H$, such that $GH^T=0$ and $\rank H=n-k$.

Given an $F$-linear code $\cal C$ and an index set ${\cal A}$ of size
$n'\equiv|{\cal A}|$, we define the punctured code
${\cal C}_{\rm p}({\cal A})$ and shortened code
${\cal C}_{\rm s}({\cal A})$ in the standard way. Important for our
purposes is that duality is maintained as
$$
{\cal C}^\perp_{\rm p}({\cal A}) = [{\cal C}_{\rm s}({\cal A})]^\perp,\quad {\cal C}^\perp_{\rm s}({\cal A}) = [{\cal C}_{\rm p}({\cal A})]^\perp.
$$
Thus, with mutually dual matrices $G$ and $H$, the punctured matrix
$G[{\cal A}]$ is also a parity check matrix of the shortened dual code
${\cal C}^{\perp}_{\rm s}({\cal A})$, while $H[{\cal A}]$ is a check
matrix of the shortened code $\mathcal{C}_{\rm s}({\cal A})$. If we
denote $G_{\cal A}$ a generator matrix of the shortened code
$C_{\rm s}({\cal A})$, then
\begin{equation}
  H[{\cal A}]\, G_{\cal A}^T=0,\quad \rank H[{\cal A}] +\rank G_{\cal A} =|{\cal A}|.
  \label{eq:setA-duality}
\end{equation}

\subsection{Polynomial and abelian group-algebra codes}
\label{sec:ga-codes}

For a given finite field $F$ and a finite abelian group $\cal G$ of order $|{\cal G}|=n$, the group algebra $F[{\cal G}]$ is defined as the $F$-linear space of all formal sums $x\equiv \sum_{g\in \cal G}x_g g$, where $x_g\in F$. An abelian $\cal G$-code in $F^n$ is an ideal $\cal J$ in the abelian ring $F[{\cal G}]$.

Given $a\in F[\mathcal{G}]$, we define its $n\times n$ matrix representation $A\equiv \MM(a)$ by its action on group elements:
\begin{equation}
  \label{eq:group-action}
  [\MM(a)]_{\alpha,\beta}\equiv \sum_{g\in \cal G}a_g\delta_{\alpha,g\beta}.
\end{equation}

\noindent{\bf Circulant matrices}: In the cyclic case
$\mathcal{G} = C_n$, $F[C_n]$ is isomorphic to the quotient ring
$R=F[x]/(x^n-1)$. A cyclic polynomial code ${\cal C}_{g(x)}$ is an
ideal in $R$ generated by $g(x)$. Its generator matrix $G$ satisfies
$G^T=\MM\biglb(g(x)\bigrb)=g(P_n)$, where $P_n$ is the order-$n$
cyclic shift matrix. For convenience in mapping to integer lattices,
we will utilize the ring of Laurent polynomials $F[x,x^{-1}]$, where
the reciprocal polynomial is denoted $\widetilde{a}(x) \equiv
a(1/x)$. Matrix transposition naturally corresponds to taking this
reciprocal: $[\MM(a)]^T = \MM(\widetilde{a})$.

\noindent{\bf Bivariate group matrices}: Consider a finite abelian
group with a two-generator presentation:
\begin{equation}
{\cal G}=\langle x,y\,|\, r_1(x,y)=r_2(x,y)=[x,y]=1\rangle,\label{eq:two-gen}
\end{equation}
where $r_j(x,y)\equiv x^{\Delta^{(j)}_{1}}y^{\Delta^{(j)}_{2}}$ and
$[x,y]\equiv xyx^{-1}y^{-1}$. Elements of the corresponding free
abelian group are in a one-to-one correspondence with the points of
the integer plane $\mathbb{Z}^2$. The exponents of $r_i(x,y)$ form
basis columns of a lattice matrix $\hat\Delta$. The group elements
correspond to inequivalent points in $\mathbb{Z}^2$ forming a torus
${\cal T}_2$ with size $n=\bigl|\det\hat\Delta\bigr|$. A presentation
where $r_i=x_i^{\Delta_i}$ corresponds to a direct product
${\cal G}=C_{\Delta_1}\times C_{\Delta_2}$, giving a torus with
periodicity vectors along the Cartesian axes.

The support of a bivariate Laurent polynomial
$a(x,y)\in F[ x^{\pm1},y^{\pm1}]$ maps to a set of points on
$\mathbb{Z}^2$. Multiplication by $x$ or $y$ translates along the
respective axis, while adding periodicity vectors preserves relative
locations modulo the quotient ring relations. In this bivariate case,
the matrix transposition rule extends as
$[\MM(a)]^T = \MM(\widetilde{a})$, where
$\widetilde{a}(x,y) \equiv a(x^{-1}, y^{-1})$.

\subsection{Quantum CSS codes}
\label{sec:CSS}

A quantum Calderbank-Shor-Steane (CSS)
code\cite{Calderbank-Shor-1996,Steane-1996} ${\cal Q}=\css(H_X,H_Z)$
over a finite field $F$, with parameters $[[n,k,d_X/d_Z]]$ or
$[[n,k,d]]$, where $d=\min(d_X,d_Z)$, is constructed from a pair of
$n$-column matrices over $F$ with mutually orthogonal rows,
\begin{equation}
  \label{eq:orthogonality-CSS} H_X H_Z^T=0.
\end{equation} Such a code is isomorphic to a direct sum ${\cal
  Q}={\cal Q}_X\oplus {\cal Q}_Z={\cal C}_{H_Z}^\perp/{\cal
  C}_{H_X}\oplus {\cal C}_{H_X}^\perp/{\cal C}_{H_Z}$. A non-trivial
codeword $\bs c\in \mathcal{Q}_Z$ satisfies
\begin{equation}
  \label{eq:CSS-cz} H_X \bs c=0,\quad \bs c\neq H_Z^T\bs \alpha,\;
\forall \bs \alpha\in F^r,
\end{equation} where $r$ is the number of rows in $H_Z$. The quantum
code dimension is
\begin{equation}
  \label{eq:k-CSS} k=n-\rank H_X-\rank H_Z,
\end{equation} and the CSS distances are the minimum weights of
non-trivial vectors in ${\cal C}_{H_Z}^\perp$ and ${\cal
  C}_{H_X}^\perp$, respectively:
\begin{equation}
  \label{eq:d-CSS} d_X=\min_{\bs c\in \mathcal{C}_{H_Z}^\perp\setminus
\mathcal{C}_{H_X}}\wgt \bs c,\quad d_Z=\min_{\bs c\in
\mathcal{C}_{H_X}^\perp\setminus \mathcal{C}_{H_Z}}\wgt \bs c.
\end{equation} Logical generator matrices $L_X$ and $L_Z$ of size
$k\times n$ satisfy:
\begin{equation}
  \label{eq:logical} L_XH_Z^T=0,\quad L_ZH_X^T=0,\quad \rank
(L_XL_Z^T)=k.
\end{equation}

Physically, the elements of ${\cal C}_{H_X}$ and ${\cal C}_{H_Z}$ map
to Pauli $X$ and $Z$ operators that generate an abelian stabilizer
group ${\cal S}$; taking these generators with a negative sign defines
a commuting Hamiltonian whose ground state manifold coincides with the
code space.

\subsection{CSS code as a chain complex}

Quantum Galois-qudit CSS codes over a finite field $F$ naturally
correspond to bounded chain complexes of finite-dimensional vector
spaces over $F$. We define a $D$-complex
${\cal A}\equiv {\cal K}(A_1,\ldots,A_D)$ in terms of
$n_{j-1}\times n_j$ matrices $A_j$ over $F$ serving as boundary
operators, with fixed-basis spaces $\mathcal{A}_j$ of dimension $n_j$:
\begin{equation}
  \label{eq:chain-complex}
  {\cal A}:\;
  \ldots \leftarrow\{0\}\stackrel{\partial_0}\leftarrow {\cal A}_0\stackrel{A_1}\leftarrow {\cal A}_1\ldots \stackrel{A_{D}}\leftarrow {\cal A}_{D}\stackrel{\partial_{D+1}}\leftarrow \{0\} \ldots
\end{equation}
where $A_{j-1}A_{j}=0$ for $j\in\{2,\ldots,D\}$. The $j$-th homology group is defined as
$$H_j({\cal A})\equiv H(A_j,A_{j+1}) = \ker(A_{j})/\im(A_{j+1}),$$
with rank:
\begin{equation}
  \label{eq:homo-rank}
  k_j\equiv \rank H_j(\mathcal{A})=n_j-\rank A_j-\rank A_{j+1}  .
\end{equation}

The co-chain complex $\widetilde{\cal A}$, formed by the transposed
matrices $A_j^T$ in reverse order, yields the co-homology group
$\widetilde{H}_j(\widetilde{\cal A})=H(A_{j+1}^T,A_j^T)$ of the same rank.

A quantum CSS code with stabilizer generator matrices $H_X=A_j$ and
$H_Z=A_{j+1}^T$ is isomorphic to the direct sum of these groups:
\begin{equation}
  \label{eq:css-code-homology}
\css(A_j,A_{j+1}^T)\cong H(A_j,A_{j+1})\oplus H(A_{j+1}^T,A_j^T).
\end{equation}
The two terms correspond to $Z$ and $X$ logical operators, respectively.

\subsection{Two-block codes}\label{sec:bgCSS}

Quantum two-block codes\cite{Kovalev-Pryadko-Hyperbicycle-2013} are CSS codes defined in terms of two commuting square matrices over $F$,
\begin{equation}
\label{HxHz}
    H_{X}=
    \begin{pmatrix}
      A & B \\
    \end{pmatrix}, \quad
    H_{Z}^T=
    \begin{pmatrix}
      B\\
      -A
    \end{pmatrix}.
\end{equation}
We focus on special cases that map to vertex-transitive planar codes: generalized-bicycle (GB) codes, bivariate-bicycle (BB) codes, and translation-invariant HP codes. In all of these cases, the group algebra elements are polynomials $a,b\in F[x^{\pm 1},y^{\pm 1}]$, with commuting variables subject to the group relators $r_1(x,y)=r_2(x,y)=1$. A non-trivial $Z$-codeword $\bs c=[\bs u,\bs v]$ corresponds to a pair of polynomials $(u,v)$ satisfying the CSS equations:
\begin{eqnarray}
  \label{eq:CSS-poly1}
  0&=&\left(a,b\right) \begin{pmatrix}u\\ v
  \end{pmatrix}
  \bmod \biglb(r_1-1,r_2-1\bigrb),\\
  \label{eq:CSS-poly2}
\begin{pmatrix}  u\\v
\end{pmatrix}
   &\neq& \alpha\begin{pmatrix}b\\ -a
   \end{pmatrix}
  \bmod \biglb(r_1-1,r_2-1\bigrb),
\end{eqnarray}
where $\alpha\equiv \alpha(x,y)$ is an arbitrary polynomial.

These code families differ primarily by their group presentation. GB codes rely on a single-variable cyclic shift representation $A=a(P_\ell), B=b(P_\ell)$. BB codes are quasi-abelian codes over a group with a two-generator presentation, while HP codes emerge when the group is a direct product of two cyclic groups, allowing $A$ and $B$ to decompose into independent variables $a=a(x)$ and $b=b(y)$.

\noindent{\bf Graphical representation}: For BB codes, the elements of the free abelian group naturally map to points on the integer plane $\mathbb{Z}^2$. As illustrated in Fig.~\ref{fig:planar}, we identify the rows of $H_X$ with vertices of the square lattice, the left and right matrix blocks with horizontal and vertical edges, and the rows of $H_Z$ with the plaquettes. 

\begin{figure}[htbp]
    \centering
    \includegraphics[width=1.0\columnwidth]{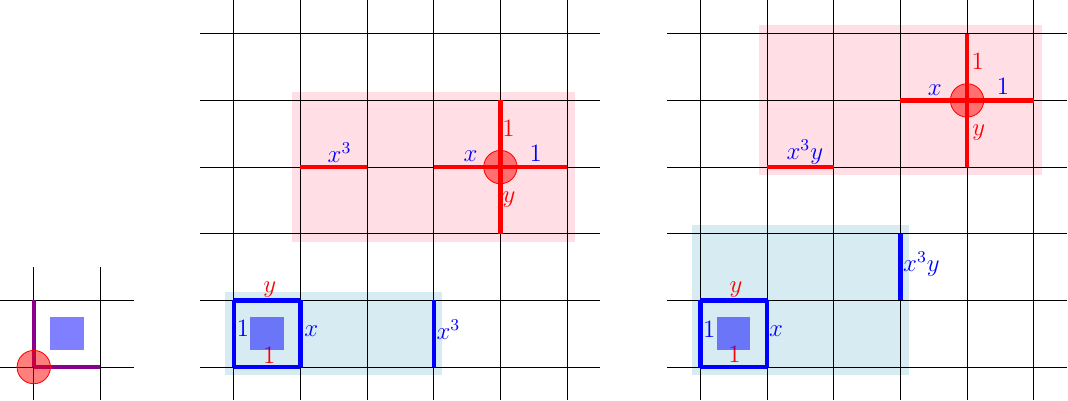}
    \caption{Planar layout of BB codes. Left: Marked vertex, two
      edges, and shaded plaquette correspond to the same lattice
      coordinates. Center: Red and blue edge patterns respectively
      correspond to $X$ and $Z$ stabilizer generators of a
      vertex-transitive HP code obtained from polynomials
      {\color{blue}$a(x,y)=1+x+x^3$}, {\color{red}$b(x,y)=1+y$}; edges
      corresponding to the monomials are labeled with matching colors,
      blue for $a$ and red for $b$. Right: Same for a BB code from
      polynomials {\color{blue}$a(x,y)=1+x+x^3y$},
      {\color{red}$b(x,y)=1+y$}.}
    \label{fig:planar}
\end{figure}

\noindent{\bf Equivalences and automorphisms}: Abelian two-block
group-algebra codes possess numerous equivalent forms that alter this
graphical representation. These arise from rescaling transformations
\begin{equation}
  a\to \alpha g a \text{ and } b\to \beta h b,\label{eq:translation}
\end{equation}
for arbitrary field elements $\alpha,\beta\in F$ and monomials
$g, h\in {\cal G}$, where the latter correspond to independent
sublattice translations. Furthermore, group automorphisms correspond
to invertible monomial substitutions. For a two-generator group, this
amounts to a unimodular transformation with integer exponents,
$(x, y) \mapsto (x^m y^n, x^p y^q)$ where $|mq - np| = 1$. Such
substitutions allow us to shift polynomials into canonical forms
without altering the code parameters.

Finally, under CSS duality ($H_X\leftrightarrow H_Z$), the two-block matrices transform as $A\to B^T$ and $B\to -A^T$. Using the transposition rule $[\MM(a)]^T = \MM(\widetilde{a})$, this yields an equivalent code constructed from the reciprocal polynomials:
\begin{equation}
  \label{eq:duality-CSS}
 a'=\widetilde{b},\quad  b'=-\widetilde{a}.
\end{equation}

\section{Two-block chain complex on the infinite plane}
\label{sec:covering-plane}

In a toric or a surface code, any non-trivial codeword must either be
a homologically non-trivial cycle or connect two boundaries. To ensure
a similar property for families of BB or BBS codes, we formulate the
corresponding condition on the infinite plane, the covering manifold
for all such codes. Namely, we associate a family of BBS codes with a
pair of bivariate Laurent polynomials $a\equiv a(x,y)$ and
$b\equiv b(x,y)$ in $F[x^{\pm1},y^{\pm1}]$, defining the {\em parent
  chain complex}
\begin{equation}
  {\cal P}\equiv   {\cal P}(a,b)\equiv
  \{0\}\stackrel{\partial_0}\leftarrow {\cal P}_0
  \stackrel{\partial_1}\leftarrow {\cal P}_1
  \stackrel{\partial_2}\leftarrow {\cal P}_2
  \stackrel{\partial_3}
  \leftarrow\{0\}.
  \label{eq:cover-complex}
\end{equation}
Here, the $F$-linear vector spaces ${\cal P}_0$,
${\cal P}_1=\mathcal{U}\oplus\mathcal{V}$, and ${\cal P}_2$ are
indexed by vertices, horizontal and vertical edges, and plaquettes of
the infinite square lattice $\mathbb{Z}^2$, with elements of vector
spaces naturally represented by Laurent polynomials.

The boundary operators $\partial_0$ and $\partial_3$ are trivial,
$\ker(\partial_0)=\mathcal{P}_0$ and $\im(\partial_3)=\{\bs 0\}$,
while the remaining two boundary operators act as
\begin{eqnarray}
  \label{eq:boundary-op1}
  &\partial_1:& t(x,y)=\biglb( a(x,y),b(x,y)\bigrb)\begin{pmatrix} u(x,y)\\ v(x,y)
  \end{pmatrix}, \\
  \label{eq:boundary-op2}
  &\partial_2:& \begin{pmatrix} u(x,y)\\v(x,y)  \end{pmatrix}
  = \begin{pmatrix}b(x,y)\\ -a(x,y)  \end{pmatrix} w(x,y).
\end{eqnarray}
Pairs of polynomials generated by $\partial_2$ correspond to elements
of the parent stabilizer group $\overline{\cal S}_Z$. The stabilizer
group $\overline{\cal S}_X$ is obtained from the co-chain complex
$\widetilde{\cal P}$ dual to $\mathcal{P}$, with elements
corresponding to polynomials of the form
$\xi(x,y)[\widetilde{a}(x,y), \widetilde{b}(x,y)]^T$. Just as in the
finite case (see Sec.~\ref{sec:bgCSS}), equivalent parent complexes
can be obtained via sublattice permutations, sublattice translations, and
automorphisms of $\mathbb{Z}^2$.

The condition for the absence of bulk codewords—meaning the complex is
exact—is guaranteed by the following Lemma, formulated for a free
abelian group with $D$ generators:

\begin{lemma}[The bulk condition]
  \label{th:bulk}
  Let $F$ be a finite field and
  $a(\bs x),b(\bs x)\in R\equiv F[\bs x, \bs x^{-1}]$ a pair of
  Laurent polynomials over $D$ variables,
  $\bs x\equiv (x_1,x_2,\ldots, x_D)$. Any pair of polynomials
  $u(\bs x), v(\bs x)\in R$ solving the equation
  $a(\bs x) u(\bs x)+b(\bs x) v(\bs x)=0$ can be written in the form
  \begin{equation}
    \begin{pmatrix}u(\bs x)\\ v(\bs x)
    \end{pmatrix}
    =\xi(\bs x) \begin{pmatrix}b(\bs x)\\ -a(\bs x) 
    \end{pmatrix},\quad
  \xi(\bs x)\in R,\label{eq:degeneracy}
\end{equation}
if and only if $a(\bs x)$ and $b(\bs x)$ have no common non-unit
factors.
\end{lemma}
This follows directly from the fact that the ring of multivariate
polynomials $F[\bs x]$ is a unique factorization domain (UFD).

Under the conditions of Lemma \ref{th:bulk}, the covering chain
complex (\ref{eq:cover-complex}) is exact, meaning it contains no
finite-weight non-trivial vectors. In the context of topological order
(TO), the parent complex (\ref{eq:cover-complex}) defines a
translationally-invariant Hamiltonian in 2D. The bulk condition in
Lemma~\ref{th:bulk} with $D=2$ is precisely the
condition\cite{Haah-2013,Liang-Xu-Iosue-Chen-2024} for TO, formulated
specifically for BB codes by Liang et
al.\cite{Liang-Liu-Song-Chen-2025}, who also give the total quantum
dimension associated with the TO in the ground state manifold on
the infinite plane as the dimension of a quotient space:
\begin{equation}
  k_{\rm TO} = \dim_F\frac{F[x^{\pm 1},y^{\pm 1}]}{ \langle a,b\rangle}. \label{eq:BBS-k}
\end{equation}
For a finite result, the ideal generated by $a$ and $b$ must be
zero-dimensional, which is equivalent to the condition in
Lemma~\ref{th:bulk} that the polynomials are mutually prime. To
compute this dimension in practice, one maps the problem to the
standard polynomial ring $F[x,y]$, where the dimension of the quotient
space $F[x,y]/\langle a,b\rangle $ can be determined using a Gr\"obner
basis. This dimension equals the total number of joint roots in the
extension field. However, because our parent complex operates over the
Laurent ring $F[x^{\pm1},y^{\pm1}]$---where the variables represent
invertible lattice translations---any roots at zero or infinity are
strictly disallowed. Therefore, the dimension
$k_{\rm TO}$  corresponds
only to the number of non-zero, finite joint roots, counted with their
multiplicity (see Ref.~\onlinecite{Haah-2013} and Appendix
\ref{app:root-count} for a precise formulation in the present
context).

\subsection{Single-sublattice patterns}
\label{sec:one-sublattice}

Given a pair of polynomials $a,b\in F[x^{\pm1},y^{\pm1}]$ with
$a\neq0$, consider a cycle in the parent chain complex
$\mathcal{P}(a,b)$ restricted to horizontal edges. This amounts to
solving Eq.~(\ref{eq:boundary-op1}) with $t=v=0$, yielding:
\begin{equation} a(x,y) u(x,y)=0.\label{eq:u-chain}
\end{equation}
While this equation admits no non-zero solutions in
$F[x^{\pm1},y^{\pm1}]$, it can be resolved in every point of
$\mathbb{Z}^2$ to yield infinite-weight vectors.  Because $a\neq0$, any
non-zero locally finite solution is automatically non-trivial.

The geometrical structure of these solutions is strictly dictated by
the support of $a(x,y)$. If the monomials in $a(x,y)$ are
\emph{aligned} along a single direction---meaning $a(x,y)$ can be
expressed as a univariate polynomial $a_1(t)$ in some combined
variable $t=x^i y^j$---the solutions to Eq.~(\ref{eq:u-chain}) are
one-dimensional, infinite-weight periodic chains. The number of
linearly independent infinite chain solutions is given exactly by the
\emph{degree spread} of the polynomial,
$\Delta_{\rm deg} (a_1)\equiv \deg_{\rm max}(a_1)-\deg_{\rm
  min}(a_1)$, which is also the number of finite non-zero roots of
$a_1$.

Conversely, when the monomials in $a(x,y)$ are not aligned,
non-trivial solutions form two-dimensional
fractals\cite{Yoshida-annals-2013,Yoshida-2013,Kalachev-Panteleev-2020,Symons-Rajput-Browne-2025}. For
instance, a binary polynomial $a=1+x+y$ generates a Sierpi\'nski
triangle pattern. The net weight of such a pattern in a ball of radius
$\rho$ scales as $\mathcal{O}(\rho^\delta)$, where $\delta$ is the
fractal dimension, yielding an upper distance
bound\cite{Kalachev-Panteleev-2020} $d\le
\mathcal{O}(n^{\delta/2})$. For sufficiently long codes, these fractal
bounds are superseded by the Bravyi-Terhal
bound\cite{Bravyi-Terhal-2009} arising from finite-width chain
solutions, which we analyze next.

\subsection{Horizontal finite-width chain solutions}
\label{sec:strip-cartesian}

Consider a pair of mutually prime polynomials
$a,b\in F[x^{\pm1},y^{\pm1}]$ that satisfy the bulk condition in Lemma
\ref{th:bulk}. Rescale them to contain only non-negative $x$ and $y$
degrees with non-zero constant terms $a(0,0)\neq0$, $b(0,0)\neq0$. We
define the maximum $y$-degree $\Delta\equiv \max(\deg_ya, \deg_yb)$,
and assume without loss of generality that
$\Delta=\deg_y a\ge \deg_yb$.

To construct a finite-width $Z$-chain fitting within a horizontal
strip of width $m$, we expand the boundary equations in powers of $y$,
taking coefficients in the univariate polynomial ring $F[x]$. Ordering
the variables $u_j,v_j\in F[x]$ as
$\bs c\equiv[u_0,v_0,\ldots,u_{m-1},v_{m-1}]^T$, the condition for a
non-trivial chain is:
\begin{equation}
  \label{eq:strip-general}
  {\sf M}_X \bs c=0,\quad \bs c\neq  {\sf M}_Z^T\bs \alpha,
\end{equation}
where the banded CSS matrices over $F[x]$ are given by:
  \begin{eqnarray}
    \label{eq:banded-H1}
    {\sf M}_X
    &=&\left(
        \def\ad{a_{\Delta}}
        \def\bd{b_{\Delta}}
        \begin{array}[c]{cc|cc|c|cc|cc}
          a_0&b_0&   \\ \cline{3-4}
          a_1&b_1&a_0&b_0 \\
          :  &:  &a_1&b_1 \\ \cline{6-7}
          \ad&\bd&: &:& \;{\scriptscriptstyle{\boldsymbol\ddots}}\;\strut &a_0 &b_0  \\
          \cline{1-2}\cline{8-9}
             &   & \ad&\bd&&a_1 &b_1&a_0 &b_0 \\
          \cline{3-4}
             &   &   &   &   & : &:  &a_1  &b_1   \\ 
             &   &   &   &   &\ad&\bd & :&: \\ \cline{6-7}
             &   &   &   &   &   &   &   \ad&\bd\\
        \end{array}\right),\hbox{\qquad}
  \end{eqnarray}
  \begin{eqnarray}
    \label{eq:banded-H2}
    {\sf M}_Z^T
    &=& \left(
        \begin{array}[c]{cccccc}
          b_0     & \\
          - a_0& \\ \cline{2-2}
          b_1& b_0 \\
          - a_1&- a_0&\\ 
          {\scriptscriptstyle\boldsymbol\vdots}& b_1& {\scriptscriptstyle\boldsymbol\ddots} \\ \cline{4-4}
          {\scriptscriptstyle\boldsymbol\vdots}&- a_1 & & b_0\\
          b_{\Delta} &         {\scriptscriptstyle\boldsymbol\vdots}& &- a_0\\
  
          - a_{\Delta}&      {\scriptscriptstyle\boldsymbol\vdots}& &b_1\\         \cline{1-1}
                  & b_{\Delta} &&- a_1\\     

                  & - a_{\Delta}&
             &{\scriptscriptstyle\boldsymbol\vdots}\\ \cline{2-2}
                  & & {\scriptscriptstyle\boldsymbol\ddots}
             &
               {\scriptscriptstyle\boldsymbol\vdots}\\
                  & &      & b_{\Delta}\\
                  & &  & - a_{\Delta}\\
         \end{array}\right),
\end{eqnarray}
with $b_\Delta=0$ if $\deg_y(b)<\Delta$. These matrices satisfy
${\sf M}_X{\sf M}_Z^T=0$, defining a chain complex over $F[x]$.

Because $F[x]$ is a principal ideal domain, we can apply a basis
transformation to the intermediate vector space that simultaneously
brings both boundary operators into Smith normal form (SNF) while
explicitly preserving their orthogonality (${\sf M}_X{\sf
  M}_Z^T=0$). This ensures that for each aligned column index $j$, the
corresponding SNF invariants $\delta_j(x)$ of ${\sf M}_X$ and
$\mu_j(x)$ of ${\sf M}_Z$ cannot be simultaneously non-zero.

This mutually exclusive pairing strictly dictates the supported
solutions. A non-unit invariant $\delta_j(x) \neq 0$ (which implies
$\mu_j=0$) with a degree spread
$\kappa_j' \equiv \Delta_{\rm deg} \delta_j(x) > 0$ generates exactly
$\kappa_j'$ linearly independent infinite-weight periodic chains and
$\kappa_j'$ finite-weight polynomial co-chains. Conversely, the dual
case ($\delta_j=0, \mu_j\neq 0$) yields finite chains and infinite
co-chains.

By variable counting, a non-trivial solution generically exists for
$m\ge \Delta$. When $m=\Delta$, ${\sf M}_Z$ is empty (meaning
$\mu_j=0$ for all $j$), and ${\sf M}_X$ is a square Sylvester
matrix\cite{Sylvester-1853} whose determinant---the product of its SNF
invariants $\delta_j(x)$---is proportional to the $y$-resultant of the
two polynomials:
\begin{equation}
  \det {\sf M}_X^{(\Delta)}=\pm a_\Delta(x)^{\Delta-\deg_y b}\res_y(a,b).
  \label{eq:det-Mx}
\end{equation}
This establishes a direct link between the resultant and the supported
solutions. Assuming $a_\Delta(x)$ is unit, the degree spread of the
resultant, $\kappa = \Delta_{\rm deg} \res_y(a,b)$, is exactly the sum
of the degree spreads of the individual invariants $\delta_j(x)$ of
${\sf M}_X$. Thus, $\kappa$ precisely counts the number of independent
infinite periodic chains (and dual polynomial co-chains) accommodated
by the strip. Furthermore, $\det {\sf M}_X^{(\Delta)}=0$ if and only
if $a$ and $b$ share a common factor.  This is consistent with the
absence of finite-weight bulk codewords being equivalent to the bulk
condition (Lemma \ref{th:bulk}).

For $m > \Delta$, we classify the invariant factors of these boundary
operators via the following statements:

\begin{statement}
  \label{th:SNF-Mx}
  Assuming polynomials $a,b\in F[x,y]$ are mutually prime,
  not identically zero at $y=0$, and have $y$ degrees
  $\Delta\equiv \deg_y a \ge \deg_y b$, the matrix ${\sf M}_X$ has
  exactly $m+\Delta$ non-zero SNF invariants. If, in addition, we
  assume one of the edge conditions,
  \begin{equation}
    \gcd\biglb(a_0(x),b_0(x)\bigrb)=1\;\,\text{or}\;\,
    \gcd\biglb(a_\Delta(x),b_\Delta(x)\bigrb)=1,\label{eq:SNF-Mx}
  \end{equation}
  then the first $m-\Delta$ SNF invariants equal $1$, while the
  remaining $2\Delta$ invariant factors identically coincide with
  those of the square matrix ${\sf M}_X^{(\Delta)}$.
\end{statement}

\begin{statement}
  \label{th:SNF-Mz}
  For any $m>\Delta$, assuming $a\neq0$ or $b\neq0$, the matrix
  ${\sf M}_Z$ has exactly $m-\Delta$ non-zero SNF invariants.  If, in
  addition, polynomials $a,b\in F[x,y]$ are mutually prime, these SNF
  invariants are all equal to $1$.
\end{statement}

To clarify the role of the factor $a_\Delta(x)$ in the determinant of
the matrix ${\sf M}_X$ when $m=\Delta$ [Eq.~(\ref{eq:det-Mx})]: this
extra term appears when the upper boundary is ``unbalanced,'' meaning
$\deg_y b < \Delta$.  If we resolve this by dropping appropriate
qubits near the top edge, the matrix reduces exactly to a Sylvester
matrix. In this case, we have
$\kappa' \equiv \Delta_{\rm deg} \res_y(a,b)$, the degree spread of
the $y$-resultant of the two polynomials.  Furthermore, provided
either of the edge conditions in Statement~\ref{th:SNF-Mx} is
satisfied, exactly this same number of solutions, $\kappa'$, exists in
a strip of any width $m>\Delta$, as long as the upper boundary stays
similarly balanced.

\section{Two-block chain complex with boundaries}
\label{sec:with-boundaries}

\subsection{Horizontal strips and edge-localized solutions}
\label{sec:strip-smooth-rough}

In a finite CSS code, boundaries can be added by reducing the code to
a subset of qubits. If we denote such a subset as region
$\mathcal{A}$, a smooth boundary is formed by puncturing the $X$
stabilizer group, $S_X\to S_X[{\cal A}]$, and shortening the $Z$
stabilizer group, $S_Z\to (S_Z)_{\cal A}$, while a rough boundary is
defined by the exact CSS-dual procedure. Shortening an abelian group
requires row operations on the corresponding generator matrix to find
generators supported strictly within $\mathcal{A}$. However, for our
translationally invariant parent complex on the infinite plane, with
region ${\cal A}$ an infinite strip, the generators that would
normally require these row operations can be pushed to infinity. Thus,
for a \emph{bare} (or uncorrected) boundary on the infinite strip, it
is sufficient instead of the shortening, to just drop any generators
that cross the boundary.

In particular, for a strip of width $m \ge \Delta$ with two bare
smooth boundaries, the generators are given exactly by the rows of
matrices ${\sf M}_X$ and ${\sf M}_Z$ in Eqs.~(\ref{eq:banded-H1}) and
(\ref{eq:banded-H2}). As established in Section
\ref{sec:strip-cartesian}, if the polynomials have no common factors,
this strip supports $\kappa'$ independent infinite-weight horizontal
chains and an equal number of polynomial co-chains.

To isolate and study finite-weight $X$ co-chains localized strictly
next to an edge, we consider a strip with mixed boundary conditions: a
smooth lower edge and a rough upper edge. Indeed, shortening the $X$
generators at the top also removes any $X$ co-chains that cross the
strip top-to-bottom. Thus, any surviving non-trivial $X$ co-chains
must be localized strictly near the lower smooth edge.

The corresponding CSS matrices over $F[x]$ for this mixed-boundary
strip of width $m>\Delta$ are obtained by (simplified) shortening near
the top row and puncturing near the bottom row for ${\sf M}_X$, and
vice versa for ${\sf M}_Z$: {\small
  \begin{eqnarray} 
     \label{eq:banded-mix-M1} {\sf M}_X &=&\left( \def\ad{a_{\Delta}}
\def\bd{b_{\Delta}}
        \begin{array}[c]{cc|cc|ccc|cc|cc} a_0&b_0& \\ \cline{3-4}
a_1&b_1&a_0&b_0 \\ : &: &a_1&b_1 \\ \cline{6-7} \ad&\bd&: &:&
\;{\scriptscriptstyle{\boldsymbol\ddots}}\;\strut &a_0 &b_0 \\
\cline{1-2}\cline{8-9} & & \ad&\bd&\ldots &a_1 &b_1&a_0 &b_0 \\
\cline{3-4}\cline{10-11} & & & & \ldots & a_2 &b_2 &a_1 &b_1&a_0&b_0
        \end{array}\right),\quad\;\;\;\\
    \label{eq:banded-mix-M2} {\sf M}_Z^T &=& \left(
        \begin{array}[c]{cccccccccc} b_0 & \\ - a_0& \\ \cline{2-2}
b_1& b_0 \\ - a_1&- a_0&\\ {\scriptscriptstyle\boldsymbol\vdots}& b_1&
{\scriptscriptstyle\boldsymbol\ddots} \\ \cline{4-4}
{\scriptscriptstyle\boldsymbol\vdots}&- a_1 & & b_0\\ b_{\Delta} &
{\scriptscriptstyle\boldsymbol\vdots}& &- a_0\\\cline{5-5}
  
          - a_{\Delta}& {\scriptscriptstyle\boldsymbol\vdots}&
&b_1&b_0\\ \cline{1-1} & b_{\Delta} &&- a_1&- a_0&\\

                  & - a_{\Delta}&
&{\scriptscriptstyle\boldsymbol\vdots}&{\scriptscriptstyle\boldsymbol\vdots}&
\\ \cline{2-2} & & {\scriptscriptstyle\boldsymbol\ddots} &
{\scriptscriptstyle\boldsymbol\vdots}&
{\scriptscriptstyle\boldsymbol\vdots} &
{\scriptscriptstyle\boldsymbol\ddots} \\ \cline{7-7} & & &
b_{\Delta}&b_{\Delta-1}& \cdots& b_0\\ & & & - a_{\Delta}&-
a_{\Delta-1}& \cdots& - a_0
         \end{array}\right).
\end{eqnarray}} With any $m\ge 1$, both matrices ${\sf M}_X$ and ${\sf
M}_Z$ have exactly $m$ rows and $2m$ columns, and satisfy ${\sf
M}_X{\sf M}_Z^T=0$. The existence of localized edge modes is directly
governed by their SNF invariants, which tie back to the joint roots of
the polynomials:

\begin{statement}
  \label{th:mix-edge-strip} With $a_0\equiv a_0(x)\neq 0$ or
$b_0\equiv b_0(x)\neq 0$, matrices ${\sf M}_X$ and ${\sf M}_Z$ have
$m$ non-zero SNF invariants. All invariants are unit if and only if
the polynomials $a_0$ and $b_0$ are mutually prime, i.e.,
  \begin{equation}
    \label{eq:edge-condition-m1}
\gcd\biglb(a_{0}(x),b_{0}(x)\bigrb)=1.
  \end{equation}
\end{statement}

When this edge condition is met, the complex is exact near the
boundary. However, if the condition is violated, the shared roots of
$a$ and $b$ at $y=0$ manifest as non-trivial finite-weight $X$
co-chains localized near the bare lower smooth boundary and, by
duality, as $Z$ chains near the upper rough boundary.  As proven in
Appendix~\ref{app:root-edge-count-proof}, for $m\ge \Delta$, the
number of independent localized solutions exactly equals the number of
common roots at $y=0$.

To suppress these non-trivial edge operators and ``correct'' the
boundary, we must systematically modify the generators along the
edge. Consider first the simplest failure of the edge condition, where
only the elements in the first row of Eq.~(\ref{eq:banded-mix-M1})
have a non-unit common factor, 
\begin{equation}
h_0(x)\equiv\gcd \biglb(a_0(x),b_0(x)\bigrb),\label{eq:h0}
\end{equation}
but the elements in
the second (and
thus all subsequent rows) have no common factors,  
$$\gcd\biglb(a_0(x),b_0(x),a_1(x),b_1(x)\bigrb) = 1.$$
This scenario corresponds to $k_0 \equiv \deg h_0(x)$ roots of 
multiplicity one in $y$.  In this case, only the first boundary  row
requires modification.  Namely, we define the reduced boundary
polynomials as:
\begin{equation} \dot a_0(x)\equiv a_0(x)/h_0(x),\quad \dot
b_0(x)\equiv b_0(x)/h_0(x).\label{eq:edge-m1-ratios}
\end{equation}
Replacing $a_0, b_0$ with $\dot a_0, \dot b_0$ in the first row of
Eq.~(\ref{eq:banded-H1}) [and Eq.~(\ref{eq:banded-mix-M1})] removes
this common factor from SNF invariants.

If the GCD of the elements in the second row is also non-unit, this
 indicates a higher-order intersection between the two
curves $a=0$ and $b=0$ at $y=0$. Each additional row with a non-unit
GCD generates further localized solutions corresponding to these
higher-multiplicity roots. By systematically canceling out the collective GCDs
in the corresponding rows, we take care of these additional roots
exactly, row by row. Because the original polynomials $a$ and $b$ are
mutually prime (the bulk condition), the intersection multiplicity is
strictly bounded. Therefore, at most $\Delta$ boundary rows will ever
need to be modified.

A symmetric analysis applies to the upper smooth edge, governed by the
roots at $y=\infty$. Overall, to guarantee exact boundaries, the
original matrix ${\sf M}_X$ in Eq.~(\ref{eq:banded-H1}) should be
replaced by a modified matrix ${\sf M}_X'$ where the outermost
boundary rows are divided by their respective common polynomial
factors:
\begin{equation}
  \label{eq:banded-H1-prime} {\sf M}_X' =\left( \def\ad{a_{\Delta}}
\def\bd{b_{\Delta}}
    \begin{array}[c]{cc|cc|c|cc|cc} \dot a_0&\dot b_0& \\ \cline{3-4}
\ddot a_1&\ddot b_1&\ddot a_0&\ddot b_0^{\strut} \\ : &: &\dddot
a_1&\dddot b_1 ^{\strut}\\ \cline{6-7} \ad&\bd&: &:&
\;{\scriptscriptstyle{\boldsymbol\ddots}}\;\strut &a_0 &b_0 \\
\cline{1-2}\cline{8-9} & & \ad&\bd&&a_1 &b_1&a_0 &b_0 \\ \cline{3-4} &
& & & & : &: &a_1^{(\Delta)} &b_1^{(\Delta)} \\ & & & & _{\strut} &
a_\Delta''& b_\Delta'' & :&: \\ \cline{6-7} & & & & & & ^{\strut} &
a_\Delta'& b_\Delta'
    \end{array}\right).
\end{equation}
Here single and multiple dot/prime accents indicate that the elements
in the corresponding row have been divided by their GCD.

\begin{statement}
  \label{th:modified-exactness}
  The modified chain complex defined by the boundary operators
  ${\sf M}_X'$ and the unmodified ${\sf M}_Z$ satisfies
  ${\sf M}_X'({\sf M}_Z)^T=0$. Furthermore, this algebraic reduction
  explicitly eliminates the common roots of $a$ and $b$ at $y=0$ (and
  symmetrically at $y=\infty$), rendering the complex entirely devoid
  of finite-weight non-trivial co-chains localized at the smooth
  boundaries.
\end{statement}

Once the boundary operators associated with zero and infinite roots
are eliminated via the modified matrix ${\sf M}_X'$, the number of
remaining independent infinite chain and polynomial co-chain solutions
supported in a strip of width $m\ge \Delta$ is given strictly by the
number of non-zero, finite joint roots of the polynomials $a$ and
$b$. This matches exactly the total quantum dimension $k_{\rm TO}$ for
the original parent complex on the infinite plane, see
Eq.~(\ref{eq:BBS-k}) in Section~\ref{sec:covering-plane}.
Furthermore, as can be verified by looking at the corresponding
matrices with mixed boundary conditions
[cf.~Eqs.~(\ref{eq:banded-mix-M1}) and (\ref{eq:banded-mix-M2})],
fixing the boundaries fully eliminates any co-chains localized near a
single boundary---all remaining non-trivial co-chains have finite
weight and necessarily connect the two edges.

Finally, we note that in the case of an \emph{unbalanced}
boundary---where the $y$-degree of $b$ is strictly smaller than
$\Delta$---spurious localized solutions can arise if $a_\Delta(x)$ is
not a unit. The modified boundary matrix ${\sf M}_X'$ derived above
systematically resolves these spurious solutions as well. This exact
algebraic reduction naturally motivates a geometric prescription to
``balance'' the boundaries by removing some edge qubits, which we
formalize in the subsequent sections.

\subsection{Tilted strips}
\label{sec:tilted-strip}

The preceding analysis extends naturally to an infinite strip running in an arbitrary lattice direction, defined by an integer normal vector ${\bf m} = (m_x, m_y)^T$ with mutually prime components:
\begin{equation}
  R({\bf m},L)\equiv \{(i,j)\in \mathbb{Z}^2: 0\le m_x i+m_y j< L\}.\label{eq:m-strip}
\end{equation}
Regarding the lattice as a free abelian group, we can always apply an
invertible, unimodular coordinate transformation to $\mathbb{Z}^2$
such that one coordinate axis runs parallel to $R$. This maps the
tilted strip to a standard horizontal strip of width $L$.

Under this transformation, the original polynomials
$a,b\in F[x^{\pm1},y^{\pm1}]$ map to
$\underline a,\underline b\in F[\xi^{\pm1},\eta^{\pm1}]$, where
$\eta=x^{m_x}y^{m_y}$ and the corresponding exponents for $\xi$ are
determined by the B\'ezout coefficients of $m_x$ and $m_y$. The
analysis of the previous subsection then applies verbatim. In
particular, the smooth boundary conditions precluding localized edge
modes (Eqs.~\ref{eq:edge-condition-m1} and \ref{eq:SNF-Mx}) become:
\begin{equation}
  \label{eq:edge-condition-transformed}
  \gcd\biglb(\underline a_0(\xi),\underline b_0(\xi)\bigrb)=1,\quad 
  \gcd\biglb(\underline a_\Delta(\xi),\underline b_\Delta(\xi)\bigrb)=1.
\end{equation}
Because the coordinate transformation changes the polynomials, a
tilted boundary may or may not require modified edge generators
regardless of whether the original Cartesian edge conditions were
satisfied.  This corresponds to the well known fact that only the
number of finite non-zero roots is conserved by the invertible
monomial transformations of the polynomials; roots at zero or infinity
may come and go.

\subsection{A strip with two rough edges}

A horizontal strip bounded by two rough edges is the CSS-dual of a
strip with two smooth edges. Mathematically, this duality corresponds
to interchanging the matrices up to a sign and using the reciprocal
polynomials: ${\sf M}_X'(a,b)={\sf M}_Z(\widetilde b,-\widetilde a)$
and ${\sf M}_Z'(a,b)={\sf M}_X(\widetilde b,-\widetilde a)$.

Consequently, when the polynomials $a,b$ satisfy the bulk condition
(Lemma~\ref{th:bulk}), a strip with two rough edges supports no
infinite-weight chains and no finite-weight co-chains. Furthermore, if
both edge conditions are satisfied (or the edges are properly modified
as in Eq.~\ref{eq:banded-H1-prime}), there are no localized chain
solutions near either rough boundary. Thus, any non-trivial polynomial
chain solution must physically connect the two rough boundaries.

We have already seen that the same statement holds for co-chain
solutions in a strip with two parallel smooth boundaries (modified if
necessary): any polynomial co-chain must connect the two smooth
boundaries. These properties precisely mirror the string operators of
conventional surface codes.

\subsection{Multiplicity of the solutions and balanced boundaries}

Let us now count the independent solutions in a horizontal infinite
strip of width $m$ with two smooth (or two rough) edges, assuming the
bulk condition (Lemma~\ref{th:bulk}) is met. We temporarily assume
that both bare edge conditions in Eq.~(\ref{eq:SNF-Mx}) are also
satisfied.

On a horizontal strip with two smooth edges, the number of independent
infinite chains (and dual non-trivial co-chains connecting the
boundaries), denoted $\kappa'$, equals the degree spread of the
determinant of ${\sf M}_X$ [Eq.~(\ref{eq:det-Mx})]. If
$\deg_y a = \deg_y b = \Delta$, then both $a_\Delta$ and $b_\Delta$
are non-zero, and $\kappa'=\Delta_{\rm deg} \res_y(a,b)$. The same
result holds when $\Delta = \deg_y a > \deg_y b > 0$, as the edge
condition implies $a_\Delta$ must be unit, rendering the additional
determinant factors trivial. Finally, if $b$ is strictly univariate,
$b(x,y)=b_0(x)$, we find:
\begin{equation}
  \label{eq:univariate-b}
  \kappa'=\deg_y (a)\,\deg_x(b_0), 
\end{equation}
as detailed in Appendix \ref{app:univariate}. (More generally, these
raw degrees are replaced by their respective degree spreads).

When the polynomials have unequal degrees, e.g.,
$\Delta=\deg_y a > \deg_y b$, a direct vertical shift using the
equivalent polynomial $b'=y^{\Delta-\deg_y b}b$ forces
$b'_\Delta(x)\neq0$ but creates a zero at the lower boundary,
$b_0'(x)=0$. This generates zero columns in the CSS matrices and
spurious localized solutions. To systematically avoid these spurious
edge modes and recover the underlying sublattice translation symmetry
of the BB codes [Eq.~(\ref{eq:translation})], we must strictly
\emph{balance} the boundaries by removing some edges and associated qubits.

Explicitly, let $\deg_{y, \rm min}(a)$ and $\deg_{y,\rm min}(b)$ be
the minimum $y$-degrees of the polynomials, with the difference:
\begin{equation}
  \delta_y^{({\rm min})}\equiv \deg_{y, \rm min}(a)-\deg_{y,\rm min}(b).\label{eq:delta-y-min}
\end{equation}
If $\delta_y^{({\rm min})}>0$, balancing the lower smooth boundary
requires removing exactly $\delta_y^{({\rm min})}$ rows of horizontal
edges. Conversely, if $\delta_y^{({\rm min})}<0$, we must remove
$|\delta_y^{({\rm min})}|$ rows of vertical edges.  Similarly, defining the
maximum degree difference:
\begin{equation}
\delta^{({\rm max})}_y\equiv \deg_{y, \rm max}(a)-\deg_{y,\rm max}(b),
\label{eq:delta-y-max}
\end{equation}
we remove $\delta^{({\rm max})}_y$ rows of horizontal edges at the
upper boundary if positive, or $|\delta^{({\rm max})}_y|$ rows of
vertical edges if negative. Balancing ensures that every qubit is
supported by both an $X$ and a $Z$ stabilizer generator.

\begin{example}[Conventional surface code]
  \label{ex:surface-code-bnd}
  For a surface code with $a=1+x$ and $b=1+y$, both polynomials have
  non-zero constant terms ($\delta_y^{\rm(min)}=0$), while their
  maximum $y$-degrees yield $\delta_y^{\rm(max)}=-1$. An upper
  horizontal smooth boundary therefore requires removing one row of
  vertical edges, reproducing the conventional smooth surface code
  boundary (Fig.~\ref{fig:surf-lattice}).
\end{example}

\begin{figure}[htbp]
  \centering
  \includegraphics[width=0.97\columnwidth]{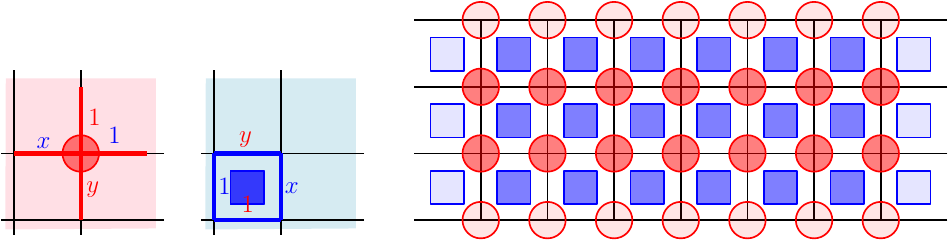}
  \caption{Structure of \emph{balanced} smooth (horizontal) and rough
    (vertical) boundaries for a surface code in Example
    \ref{ex:surface-code-bnd}.  Red and blue tiles on the left show
    stabilizer generators in the bulk.  On the right, thin black lines
    show qubit locations, placement of stabilizer generators is shown
    with blue squares and red circles, respectively, and lighter
    shading indicates that corresponding generators must be trimmed at
    the boundary.  Notice that the top row of vertical edges and the
    leftmost row of horizontal edges are removed to form balanced
    boundaries.}
  \label{fig:surf-lattice}
\end{figure}

With balanced boundaries, the horizontal edge conditions take an
invariant form:
\begin{equation}
  \label{eq:horizontal-edge-condition-inv}
  \gcd\biglb(a_{\rm min}(x),b_{\rm min}(x)\bigrb)=
  \gcd\biglb(a_{\rm max}(x),b_{\rm max}(x)\bigrb)=1.
\end{equation}
Provided the strip is wide enough to support non-trivial operators,
the number of solutions in such a balanced strip becomes independent of
$m$:\footnote{For consistency, when calculating resultants,
  we always reduce polynomials to a form with non-negative degrees and
  non-zero free terms, which guarantees the correct code dimension.}
\begin{equation}
  \label{eq:kappa-horiz}
  \kappa_{\rm horiz}=\Delta_{\rm deg} \res_y(a,b).
\end{equation}
The invariant edge conditions (\ref{eq:horizontal-edge-condition-inv})
guarantee the absence of roots at $y=0$ or $y=\infty$. Thus,
$\kappa_{\rm horiz}$ precisely counts the finite, non-zero roots. In
the case the edge conditions are violated, we may get solutions
localized at the edges, but these can be systematically suppressed
using the modified boundary matrix ${\sf M}_X'$ derived in Section
\ref{sec:strip-smooth-rough}; see Example \ref{ex:factor-min}.

\begin{example}
  \label{ex:factor-min}
  Consider binary polynomials $a=1+x$ and $b=1+x^2+y$. These polynomials have the same
  $y$-degree difference values as those for the surface code, yielding
  the identical structure of horizontal boundaries.  However, they
  violate the invariant edge condition
  (\ref{eq:horizontal-edge-condition-inv}), sharing the factor
  $\gcd(a_{\rm min},b_{\rm min})=1+x$. This generates infinitely many
  localized finite-weight $Z$-chains at the lower boundary
  (Fig.~\ref{fig:factor-min}), necessitating modified stabilizer
  generators.
\end{example}

\begin{figure}[htbp]
  \centering
  \includegraphics[width=0.9\columnwidth]{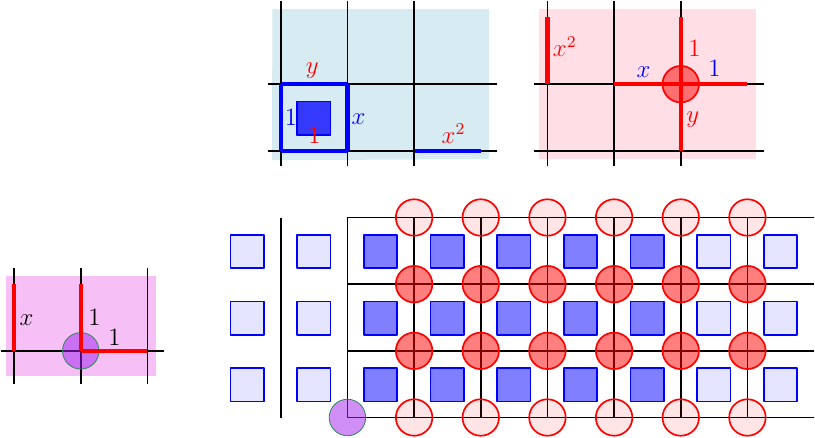}
  \caption{As in Fig.~\ref{fig:surf-lattice} but for the BBS code from
    polynomials {\color{blue}$a=1+x$} and {\color{red}$b=1+x^2+y$}
    with horizontal smooth and vertical rough boundaries.  Additional
    tile with violet shading corresponds to a non-trivial $X$ codeword
    $(u_0,v_0)=(x,1+x)$ of weight $3$ associated with the lower
    boundary.  The boundary can be fixed by adding this codeword to
    the stabilizer group in just one location as indicated with a
    purple circle, or replacing the entire bottom row of
    lightly-shaded vertex operators, which gives the same code.}
  \label{fig:factor-min}
\end{figure}

Exactly similar analysis applies to vertical strips, except now we
should expand the polynomials over $x$
with coefficients in $F[y]$.  The vertical invariant edge conditions
are:
\begin{equation}
  \label{eq:edge-vert}
  \gcd\biglb(a_{\rm min}(y),b_{\rm min}(y)\bigrb)=
  \gcd\biglb(a_{\rm max}(y),b_{\rm max}(y)\bigrb)=1,
\end{equation}
yielding a solution count $\kappa_{\rm vert}=\Delta_{\rm deg} \res_x(a,b)$. This allows us to unify the horizontal and vertical analysis:

\begin{statement}
  \label{th:horiz-vert}
  Let $a,b\in F[x^{\pm1},y^{\pm1}]$ be a pair of mutually prime
  bivariate Laurent polynomials satisfying the four invariant edge
  conditions in Eqs.~(\ref{eq:horizontal-edge-condition-inv}) and
  (\ref{eq:edge-vert}). The number of independent solutions in
  sufficiently wide, balanced horizontal and vertical strips
  coincides:
  \begin{equation}
    \label{eq:horiz-vert}
    \kappa=\Delta_{\rm deg}\res_x(a,b)=\Delta_{\rm deg}\res_y(a,b).
  \end{equation}
\end{statement}
\begin{proof}
  The bulk condition ensures the resultants are
  non-zero. Equations~(\ref{eq:kappa-horiz}) and (\ref{eq:edge-vert})
  give the counts for horizontal and vertical strips with smooth
  boundaries, respectively.  The corresponding numbers for strips with
  rough boundaries are the same, due to CSS duality.  We must simply
  show $\kappa_{\rm horiz}=\kappa_{\rm vert}$.

  The zeros of $\res_y(a,b)$ are the finite $x$-coordinates of the
  intersections between $a(x,y)=0$ and $b(x,y)=0$, while zeros of
  $\res_x(a,b)$ are the finite $y$-coordinates. The edge conditions
  strictly forbid intersections at infinity or along the axes. Thus,
  both degree spreads independently count the exact same set of
  finite, non-zero intersection points, rendering the projection
  direction irrelevant.
\end{proof}

Using the construction in Sec.~\ref{sec:tilted-strip}, any
sufficiently wide tilted strip (with properly fixed boundaries, if
necessary) carries the same number of solutions.  Consequently,
$\kappa$ is a topological invariant of the integer plane
$\mathbb{Z}^2$, equal to the total number of common finite roots with
both coordinates strictly non-zero. As shown in
Appendix~\ref{app:root-count}, this matches the total quantum
dimension $k_{\rm TO}$ associated with the TO on the infinite plane
[Eq.~(\ref{eq:BBS-k})].

\subsection{Rectangular region with two smooth and two rough
boundaries}

Having established the exact properties of linear boundaries, we can
construct a finite-size code block by intersecting two orthogonal
strips. Assuming the invariant edge conditions
(\ref{eq:horizontal-edge-condition-inv}) and (\ref{eq:edge-vert})
hold, we start with an infinite horizontal strip with balanced smooth
boundaries. We then introduce vertical rough boundaries by puncturing
the $Z$ stabilizer generators (and any infinite $Z$-chains) and
shortening the $X$ generators.

We define the physical region dimensions, $L_x$ and $L_y$, strictly in terms of horizontal edges: $L_y$ is the number of rows of horizontal edges in the original horizontal strip, and $L_x$ is the number of columns of horizontal edges in the orthogonal vertical strip. For the resulting CSS code to be non-trivial, at least one column of $X$ generators and one row of $Z$ generators must survive. The minimum required dimensions are dictated by the polynomial degree spreads:
\begin{eqnarray}
\Delta_x&\equiv& \max\biglb(\Delta_{\rm deg}^{(x)}(a), \Delta_{\rm deg}^{(x)}(b)\bigrb),
         \label{eq:min-Lx} \\
 \Delta_y&\equiv& \max\biglb(\Delta_{\rm deg}^{(y)}(a), \Delta_{\rm deg}^{(y)}(b)\bigrb).\label{eq:min-Ly}
\end{eqnarray}
As proven in Appendix \ref{app:th:code-rectang}, this construction yields our primary result:

\begin{statement}[Rectangular BBS code]
  \label{th:code-rectang}
  Under the conditions of Statement \ref{th:horiz-vert}, consider a
  rectangular region carved from a horizontal smooth strip by adding
  two vertical rough boundaries, with all edges properly balanced. For
  sufficiently large dimensions $L_x\ge \Delta_x$ and
  $L_y\ge \Delta_y$, the dimension of the resulting CSS code is
  exactly $\kappa$ [Eq.~(\ref{eq:horiz-vert})], and its minimal
  distances satisfy:
  $$d_X\ge  L_y/\Delta_y, \quad d_Z\ge L_x/\Delta_x.$$
\end{statement}

While the order of truncation (smooth first vs. rough first) dictates the precise microscopic layout of the code, neither the dimension $k$ nor the distance bounds depend on this choice. Furthermore, rotating the boundary orientations (vertical smooth and horizontal rough) preserves these parameters entirely. 

We should note that while the invariant edge conditions greatly
simplify the construction, they are not strictly mandatory. If
violated, boundaries can be algebraically modified via
Eq.~(\ref{eq:banded-H1-prime}). However, uncorrected violations along
a rough edge can produce shortened $X$ generators with unacceptably
high weights, as demonstrated in Example \ref{ex:edge-off} in
Sec.~\ref{sec:LEC-cmp}.

\section{Guaranteed BBS code construction}
\label{sec:construction}

Statement \ref{th:code-rectang} gives a practical algorithm for
constructing BBS codes. In the following, we consider the slightly more
general case of a construction based on the intersection of two tilted
strips [see Sec.~\ref{sec:tilted-strip}], $R_1\equiv R({\bf m}_1,L_1)$
with smooth and $R_2\equiv R({\bf m}_2,L_2)$ with rough boundaries,
with non-collinear normal vectors,
${\bf m}_1\times {\bf m}_2\neq {\bf 0}$. As always, we assume that
the boundaries are properly balanced for the polynomials
$a,b\in F[x^{\pm1},y^{\pm1}]$ which satisfy the bulk condition in
Lemma~\ref{th:bulk} and invariant versions of edge
conditions~(\ref{eq:edge-condition-transformed}) using the coordinates
appropriate for each strip,
\begin{equation}
  \label{eq:edge-condition-transformed-inv}
  \gcd\biglb(\underline a_{\rm min}(\xi),\underline b_{\rm min}(\xi)\bigrb)=
  \gcd\biglb(\underline a_{\rm max}(\xi),\underline b_{\rm max}(\xi)\bigrb)=1.
\end{equation}

\begin{algorithm}[{Guaranteed BBS construction}]\strut
  \label{alg:guaranteed}
\begin{enumerate}
\item Construct infinite groups ${\cal S}^X_1$ (punctured to the
  region $R_1$) and ${\cal S}^Z_1$ (shortened to $R_1$), respectively,
  using vertex generators of the parent infinite-plane complex with at
  least one edge in $R_1$, and plaquette generators with \emph{all}
  edges in $R_1$ (simplified shortening).
\item Construct groups ${\cal S}^X_2$ and ${\cal S}^Z_2$,
  respectively, by shortening ${\cal S}^X_1$ and puncturing
  ${\cal S}^Z_1$ to the intersection of the two strips, $R_1\cap R_2$.
\item Puncture out any qubits outside of the intersection of supports
  of the groups ${\cal S}^X_2$ and ${\cal S}^Z_2$, which gives the $X$
  and $Z$ stabilizer subgroups of the resulting BBS code.
\end{enumerate}
\end{algorithm}

Balanced boundaries for each region can be constructed using the
associated coordinates introduced in Sec.~\ref{sec:tilted-strip}.
Namely, the region $R_1$ can be defined using the coordinates $\xi'$, $\eta'$
associated with the normal vector ${\bf m}_1$. We include
horizontal edges $(\xi_h',\eta_h')$ and vertical edges
$(\xi_v',\eta_v')$ such that:
\begin{eqnarray}
  0 &\le& \eta_h'<L_1,\label{eq:S1h} \\
  \delta_{\eta'}^{\rm (min)}&\le& \eta_v' <L_1+\delta_{\eta'}^{\rm (max)},\label{eq:S1v}
\end{eqnarray}
where $\delta_{\eta'}^{\rm (min)}$ and $\delta_{\eta'}^{\rm (max)}$,
respectively, are the $\eta'$-degree differences for the transformed
polynomials, see Eqs.~(\ref{eq:delta-y-min}) and
(\ref{eq:delta-y-max}). The corresponding bounds for the region
$R_2$, in terms of the associated coordinates $\xi''$, $\eta''$:
\begin{eqnarray}
  0&\le& \eta''_h<L_2,\label{eq:S2h}\\
  \delta_{\eta''}^{\rm (max)}&\le& \eta''_v <L_2+\delta_{\eta''}^{\rm (min)}.\label{eq:S2v}
\end{eqnarray}
In practice, it may be more convenient to define the regions by
including all horizontal edges with untransformed coordinates $(i,j)$
such that
\begin{align}
R_1: &\qquad        0\le m_{1x}i+m_{1y}j<L_1,\tag{$\ref{eq:S1h}'$}\\
R_2: &\qquad        0\le m_{2x}i+m_{2y}j<L_2,\tag{$\ref{eq:S2h}'$}
\end{align}
and adding all vertical edges in the union of the supports of all
plaquette generators for $R_1$ and all vertex generators for $R_2$
(in both cases, all horizontal edges must fit).

Also, when constructing the group ${\cal S}_2^Z$, Steps $1$
and $2$ can be combined; we only need to include (partially or
completely) the bulk generators with plaquette positions
such that
\begin{eqnarray*}
  -\deg_{\eta'}^{\rm min}(b')&\le &\eta'_\square < L_1-\deg_{\eta'}^{\rm max}(b'),\\
  -\deg_{\eta''}^{\rm max}(b'')&\le &\eta''_\square < L_2-\deg_{\eta''}^{\rm min}(b''),
\end{eqnarray*}
where $b'$ and $b''$ are obtained from the original polynomial $b$ by
the corresponding coordinate transformations. Similarly, while
constructing the group ${\cal S}_2^X$, at Step 1 it is practical to
choose a sufficiently large buffer size $L_{\rm buf}$, and include
only the bulk generators with vertex positions
such that
\begin{eqnarray*}
  \deg_{\eta'}^{\rm max}(a')&\le &\eta'_\bullet < L_1+\deg_{\eta'}^{\rm min}(a'),\\
  -L_{\rm buf} +\deg_{\eta''}^{\rm max}(a'')&\le &\eta''_\bullet < L_2-\deg_{\eta''}^{\rm min}(a'')+L_{\rm buf}.
\end{eqnarray*}
To ensure correct code construction, it is important that we include
only the uncut bulk generators. This ensures that shortening at Step
2 does not produce, e.g., an $X$ generator corresponding to a vertical
co-chain.

In the rest of this section we compare Algorithm \ref{alg:guaranteed}
with related constructions found in the literature.

\subsection{Application to hypergraph-product codes}
Here we consider HP codes on the infinite plane, which most generally
correspond to the case of two \emph{aligned} polynomials, see
Sec.~\ref{sec:one-sublattice}. In the simplest case the polynomials
depend on one variable each, $a(x,y)=f_1(x)$ and $b(x,y)=f_2(y)$. It
is easy to verify that such polynomial pairs have no common factors
and always satisfy the edge conditions. Without limiting generality,
let us assume both polynomials have non-zero free terms,
$f_j(0)\neq 0$, and denote their degrees $k_j\equiv \deg f_j(x)$,
$j=1,2$. The common roots of $a$ and $b$ are simply all possible root
pairs $(x_i,y_j)$, with $f_1(x_i)=0$, $f_2(y_j)=0$, which gives
$\kappa=k_1k_2$ for the dimension~(\ref{eq:BBS-k}). Equivalently,
the two resultants in this case are $\res_y(a,b)=[f_1(x)]^{k_2}$ and
$\res_x(a,b)=[f_2(y)]^{k_1}$, so that Eq.~(\ref{eq:horiz-vert}) also
gives $\kappa=k_1k_2$, in agreement with Eq.~(\ref{eq:univariate-b}).

First, consider the rectangular case, an intersection of a horizontal
strip $R_1$ of width $L_1\equiv L_y$ with smooth edges and a vertical
strip $R_2$ of width $L_2\equiv L_x$ with rough edges. The minimum
degree differences (\ref{eq:delta-y-min}) are both trivial,
$\delta_x^{\rm (min)}=\delta_y^{\rm (min)}=0$, while the maximum ones
in Eq.~(\ref{eq:delta-y-max}) have opposite signs,
$\delta_x^{\rm (max)}=k_1$, $\delta_y^{\rm (max)}=-k_2$. It is easy
to verify that Algorithm \ref{alg:guaranteed} uses $L_x \times L_y$
horizontal edges, $(L_x-k_2)\times (L_y-k_1)$ vertical edges,
$(L_x-k_2)\times L_y$ independent $X$-generators (vertices), and
$L_x \times (L_y-k_1)$ independent $Z$-generators (plaquettes), which
gives a CSS code of dimension $k=k_1k_2$.

These parameters and the structure of stabilizer generators are
exactly consistent with those of an HP code with CSS matrices
\begin{equation}
\label{eq:HPHxHz}
\!\! H_{X}=\left(I_{r_1}\otimes H_{2},  H_{1}\otimes I_{r_{2}}\right)\!,\;
 H_{Z}^T=\left(
   \begin{array}[c]{c}
     H_{1}\otimes I_{n_{2}}\\
     -I_{n_{1}}\otimes H_{2} 
   \end{array}
 \right),
\end{equation}
constructed from a full-row-rank matrix $H_1$ with $r_1=L_y-k_1$ rows
and $n_1=L_y$ columns, and a full-column rank $H_2$ with $r_2=L_x$
rows and $n_2= L_x-k_2$ columns, where the rows of $H_1$ and $H_2$,
respectively, match the coefficients of the polynomials $x^{j}f_1(x)$
and $x^{j}f_2(x)$. In particular, with $f_1=f_2=1+x$, conventional
(unrotated) surface codes are recovered.

For the diagonal region boundaries, e.g., ${\bf m}_1=(1,1)$ and
${\bf m}_2=(-1,1)$, we use the coordinates $\xi'=\xi''=x$ and
$\eta'=xy$, $\eta''=y/x$, that is, $y=\eta'/\xi'=\eta''\xi''$. In
particular, for equal-degree polynomials, the corresponding codes
contain the same number of qudits as the rotated surface codes, with
the total of $n=L_1L_2$ horizontal and vertical edges, see
Fig.~\ref{fig:rotated-HP} for an example.

\begin{figure}[htbp]
  \centering
  \includegraphics[width=0.95\columnwidth]{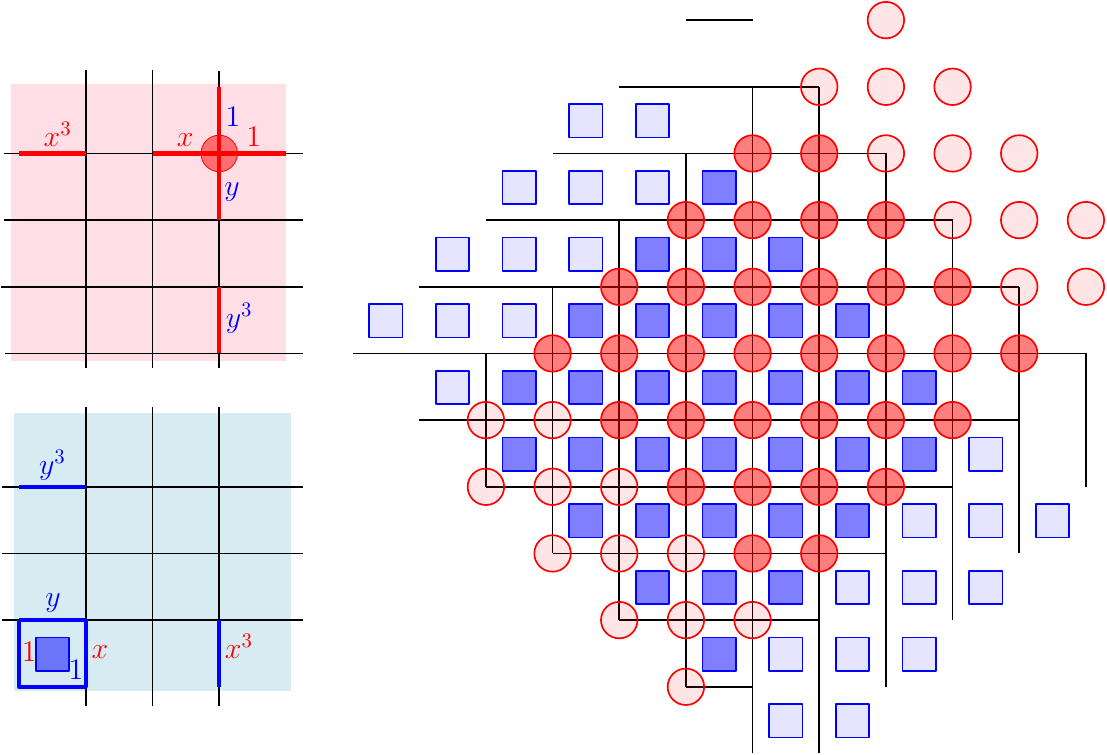}
  \caption{As in Fig.~\ref{fig:surf-lattice} but for the code
    $[[11^2,9,5]]$ produced by Algorithm \ref{alg:guaranteed} for
    equal-degree aligned polynomials
    ${\color{blue}a(x,y)}={\color{red}b(y,x)}=1+x+x^3$ with diagonal
    boundaries and $L=11$.}
\label{fig:rotated-HP}
\end{figure}

\subsection{Comparison with tile codes}

While the construction algorithm may seem different, unrotated tile
codes \cite{Steffan-etal-Eberhardt-2025,Breuckmann-etal-Steffan-2025}
are a special case of rectangular BBS codes. Indeed, a pair of
commuting red ($X$) and blue ($Z$) bulk tiles of size
$(D+1)\times(D+1)$ are readily interpreted in terms of the support of
$X$- and $Z$-generators of a BB code on the infinite plane. Here the
$x$ and $y$ degrees of individual monomials must fit the range from
$0$ to $D$, for each polynomial necessarily including the end points
of the interval.

There are several additional requirements for \emph{valid} tiles.
First of these, Eq.~(3) in
Ref.~\onlinecite{Breuckmann-etal-Steffan-2025}, requires that for
$(D+1)\times (D+1)$ tiles, each of the four corner monomials
$\{1,x^D,y^D, x^Dy^D\}$ is contained in one of the two polynomials.
Second is the condition for topological order which corresponds to the
bulk condition in Lemma \ref{th:bulk}, the necessary and sufficient
condition for the absence of non-trivial finite-weight logical
operators on $\mathbb{Z}^2$. Similarly, we interpret the requirement
of ``total topological order'' formulated implicitly in
Ref.~\onlinecite{Breuckmann-etal-Steffan-2025} as a variant of our
edge conditions (\ref{eq:SNF-Mx}).

In the construction of (rectangular) tile codes, qubits occupy a
rectangular region with equal numbers of horizontal and vertical
edges, full or partial $Z$ checks are restricted to a horizontal
strip, and full or partial $X$ checks to a vertical strip. More
precisely, the tile placement algorithm is symmetric with respect to
$X$ and $Z$ generators: full $Z$ tiles restricted to the horizontal
strip are subsequently punctured to the intersection with the vertical
strip, while full $X$ tiles restricted to the vertical strip are
punctured to the intersection with the horizontal strip. This ensures
commutativity, and away from the corners gives stabilizer generators
identical to those from Algorithm~\ref{alg:guaranteed}. A notable
advantage of tile codes is that for any pair of \emph{valid} tiles
the construction is guaranteed to give code families with the
dimension $k=2D^2$ (maximum possible for polynomials of $x$ and $y$
degrees $D$) and distances increasing with the size of the region,
without the need for any additional ``corner'' corrections.

We have examined all tile codes listed in
Ref.~\cite{Steffan-etal-Eberhardt-2025}. The corresponding
polynomials satisfy the conditions on the monomials, the bulk
condition in Lemma \ref{th:bulk}, and the edge conditions
(\ref{eq:horizontal-edge-condition-inv}), (\ref{eq:edge-vert}).
Further, the degree requirements for valid tiles ensure that balanced
boundaries are obtained without the need for removing any qubits near
the edges. We have also verified that for all polynomial pairs with
weights $w\le 4$ of equal degrees $D\le 4$ that satisfy the conditions
for valid tiles (including but not limited to the tiles listed in
Ref.~\cite{Steffan-etal-Eberhardt-2025}), the parameters of the codes 
from the tile construction agree exactly with those produced by 
Algorithm \ref{alg:guaranteed}, without the need for any additional 
corner correction.

However, for a polynomial pair that violates one of the conditions for
valid tiles, the tile construction as presented in
Ref.~\onlinecite{Steffan-etal-Eberhardt-2025} may produce a code
family with distance bounded by small-weight logical operators. As
illustrated in Example \ref{ex:tile-corner} and
Fig.~\ref{fig:tile-corner}, the tile construction in such a case may
be amended, e.g., by promoting one or a few small-weight codewords to
stabilizer generators (usually in the corners), similarly to the LEC
algorithm in Ref.~\onlinecite{Liang-Eberhardt-Chen-2025}.

\begin{figure}[htbp]
  \centering
  \includegraphics[width=3in]{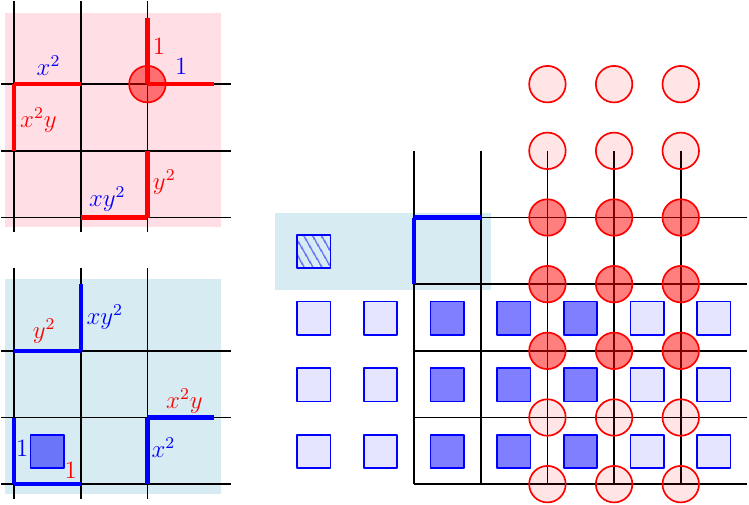}
  \caption{As in Fig.~\ref{fig:surf-lattice} but for the polynomials
    {\color{blue}$a=1+x^2+xy^2$}, {\color{red}$b=1+y^2+x^2y$} in a
    rectangular region with $n=50$ qubits ($5\times 5$ grid), see
    Example \ref{ex:tile-corner}. The tile construction gives
    $r_X=21$ (red) and $r_Z=21$ (blue) stabilizer generators shown as
    blue squares and red circles, respectively, which gives a code
    with $k=8$ and $d_Z=2$, with a minimum-weight $Z$ codeword shown
    with blue edges. The code can be amended by promoting this
    codeword to a $Z$ generator (drawn schematically as a hatched blue
    square in the corner), resulting in a $[[50,7,3]]$ code, same as
    directly given by Algorithm \ref{alg:guaranteed} with $L_x=L_y=5$.
    Increasing the code size $L_x=L_y$ gives codes with the same
    dimension and increasing distances, e.g., $[[72,7,4]]$,
    $[[98,7,5]]$, $[[128,7,6]]$, etc.}
  \label{fig:tile-corner}
\end{figure}

\begin{example}
  \label{ex:tile-corner}
  Consider binary polynomials $a=1+x^2+xy^2$, $b=1+y^2+x^2y$ of equal $x$
  and $y$ degrees $D=2$; the corresponding $3\times 3 $ tiles are
  shown in Fig.~\ref{fig:tile-corner}. This is an ``invalid'' tile pair
  since the monomial $x^2y^2$ is missing. Indeed, the tile
  construction\cite{Steffan-etal-Eberhardt-2025} gives a family of
  codes with distance $d=2$. Specifically, on a $5\times 5$ grid
  ($n=50$ qubits), the tile construction gives $r_X=r_Z=21$
  independent stabilizer generators of each type, which corresponds to
  a $k=8$ code, in agreement with the formula $k=2D^2$ for $D=2$ tile
  codes. In contrast, for these polynomials Eq.~(\ref{eq:horiz-vert})
  gives $\kappa=7$, which is also the dimension of the code produced
  by Algorithm \ref{alg:guaranteed}. The tile construction can be
  amended by adding an extra $Z$ generator, which reduces the code
  dimension to $k=7$.
\end{example}

We should also note that the condition for both polynomials to have
the same $x$ and $y$ degree spread equal to $D$ is necessary for the
tile code construction to work directly as stated, and to have the
code dimension $k=2D^2$ (Theorem 4 in
Ref.~\onlinecite{Breuckmann-etal-Steffan-2025}).  Without it, as
illustrated in Example \ref{ex:special-case} and
Fig.~\ref{fig:special-case}, this formula gives only an upper bound
for $k$.

\begin{figure}[htbp]
  \centering
  \includegraphics[width=3.0in]{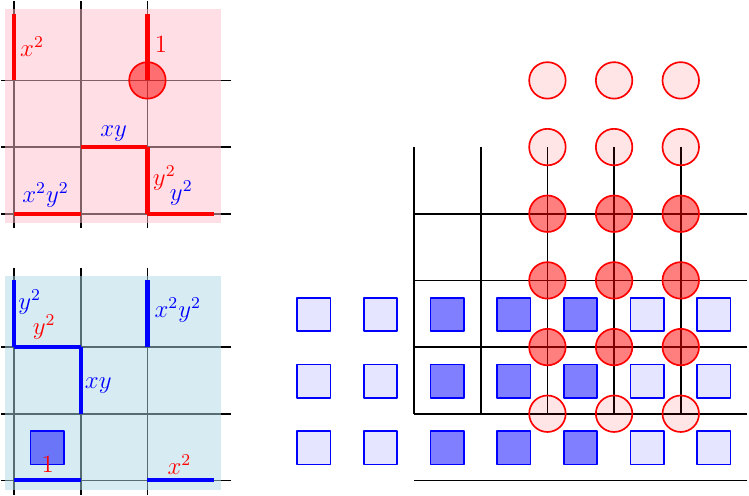}
  \caption{As in Fig.~\ref{fig:tile-corner} but for polynomials in
    Example \ref{ex:special-case}. Here no $Z$ generators are
    supported on the lowest row of vertical edges; these edges have to
    be removed along with the lowest row of $X$ generators, leaving
    $n=45$ qubits with $r_X=18$ and $r_Z=21$ independent generators,
    resulting in the code $[[45,6,3]]$. Same code is produced by
    Algorithm~\ref{alg:guaranteed} directly.}
  \label{fig:special-case}
\end{figure}

\begin{example}
  \label{ex:special-case}
  Consider binary polynomials {\color{blue}$a=y^2+xy+x^2y^2$} and
  {\color{red}$b=1+x^2+y^2$}, with the corresponding $3\times 3$ $X$
  and $Z$ tiles ($D=2$) as shown in Fig.~\ref{fig:special-case}.
  While all other conditions for valid tiles are satisfied, and the
  tile algorithm produces a code family with unbounded distances, the
  polynomial $a$ has the minimum $y$-degree $\deg_{y, \rm min}(a)=1$.
  This polynomial pair gives a code family with $k=6$, smaller than
  expected for $D=2$ tiles. Reduced dimension is due to the fact none
  of the blue tiles are supported on the vertical edges next to the
  lower boundary. As a result, for the example in
  Fig.~\ref{fig:special-case}, five qubits and three $X$ generators
  must be dropped, which gives a code with parameters $[[45,6,3]]$.
  Both Eq.~(\ref{eq:edge-vert}) and an explicit calculation with
  Algorithm \ref{alg:guaranteed} give a family of codes with dimension
  $k=\kappa=6$.
\end{example}

We conclude that our Algorithm \ref{alg:guaranteed} is strictly a
generalization of the rectangular tile
construction\cite{Steffan-etal-Eberhardt-2025,Breuckmann-etal-Steffan-2025}
(a variant of tile construction with diagonal boundaries was mentioned
in Ref.~\onlinecite{Steffan-etal-Eberhardt-2025} but was not
elaborated upon; it did not result in any codes with competitive
parameters).

\subsection{Comparison with ``open boundary'' codes \cite{Liang-Eberhardt-Chen-2025}}
\label{sec:LEC-cmp}
We have also examined the construction algorithm and corresponding
codes introduced by Liang, Eberhardt, and
Chen\cite{Liang-Eberhardt-Chen-2025}. For brevity, we refer to their
algorithm (excluding the final ``lattice grafting'' step, discussed
separately) and the resulting codes as the LEC algorithm and LEC
codes, respectively.

In our language, the LEC algorithm for constructing a family of planar
codes with open boundaries can be formulated as follows:
\begin{algorithm}[LEC algorithm]\label{alg:LEC}
  Given a pair of polynomials $a, b\in F[x^{\pm1},y^{\pm1}]$ that
  satisfy the topological order (TO) condition equivalent to that in
  Lemma \ref{th:bulk},
  \begin{enumerate}
  \item Construct ``bulk'' stabilizer generators on the infinite
    plane, see Fig.~\ref{fig:planar}
  \item
    \label{step:edge} Construct additional ``boundary'' $Z$ generators along
    vertical boundaries and $X$ generators along horizontal boundaries
    in a translation-invariant fashion, to ensure that any non-trivial
    $Z$-chain can terminate on a vertical boundary, and a non-trivial
    $X$-chain on a horizontal boundary.
  \item For a code with given
    horizontal and vertical dimensions, include the bulk generators
    inside the region and only those boundary generators that commute
    with each other (for any pair of non-commuting boundary
    generators, exclude both).
  \item
    \label{step:corner} If needed, complete the stabilizer group by adding ``corner''
    generators to get rid of local non-trivial logical operators. Do
    this by promoting $X$ or $Z$ logical operators to $X$ or $Z$
    stabilizer generators, respectively, with an ad hoc rule to give
    preference to higher-weight operators.
  \item Finally, drop any qubits which support weight-one generators.
    (Since the group is abelian by construction, this step is identical
    to Step 3 in Algorithm \ref{alg:guaranteed}.)
  \end{enumerate}
\end{algorithm}
Compared to Algorithm~\ref{alg:guaranteed}, Steps \ref{step:edge} and
\ref{step:corner} here ensure the absence of bounded-weight logical
operators without the need of the additional edge
conditions~(\ref{eq:horizontal-edge-condition-inv}) or
(\ref{eq:edge-vert}). With a sufficiently large region, the resulting
codes are expected to have the dimension equal to the quantum
dimension associated with the TO on the infinite plane, see
Eq.~(\ref{eq:BBS-k}).

We have examined the optimized code families presented in
Ref.~\onlinecite{Liang-Eberhardt-Chen-2025}. It turns out that for
all LEC code families based on weight-3 polynomials,
$\wgt a=\wgt b=3$, summarized in Table I there, all four edge
conditions (\ref{eq:horizontal-edge-condition-inv}),
(\ref{eq:edge-vert}) are satisfied, and Algorithm~\ref{alg:guaranteed}
recovers the codes with the same parameters. Moreover, for all these
codes, we have been able to construct the sets of stabilizer
generators which are simply a puncture of the corresponding bulk
generators to the set of qubits in the code. That is, for these
polynomial pairs Algorithm~\ref{alg:guaranteed}, based entirely on
linear algebra, gives exactly the same codes as constructed by the LEC
algorithm.

However, for the only LEC code family based on polynomials of weight
four (Sec IV.D of Ref.~\onlinecite{Liang-Eberhardt-Chen-2025}), one of
the edge conditions is not satisfied. We examine this case in detail
in Example~\ref{ex:edge-off} (see also the related Example~\ref{ex:edge-off-dual}):

\begin{figure}[htbp]
  \centering
  \includegraphics[width=0.95\columnwidth]{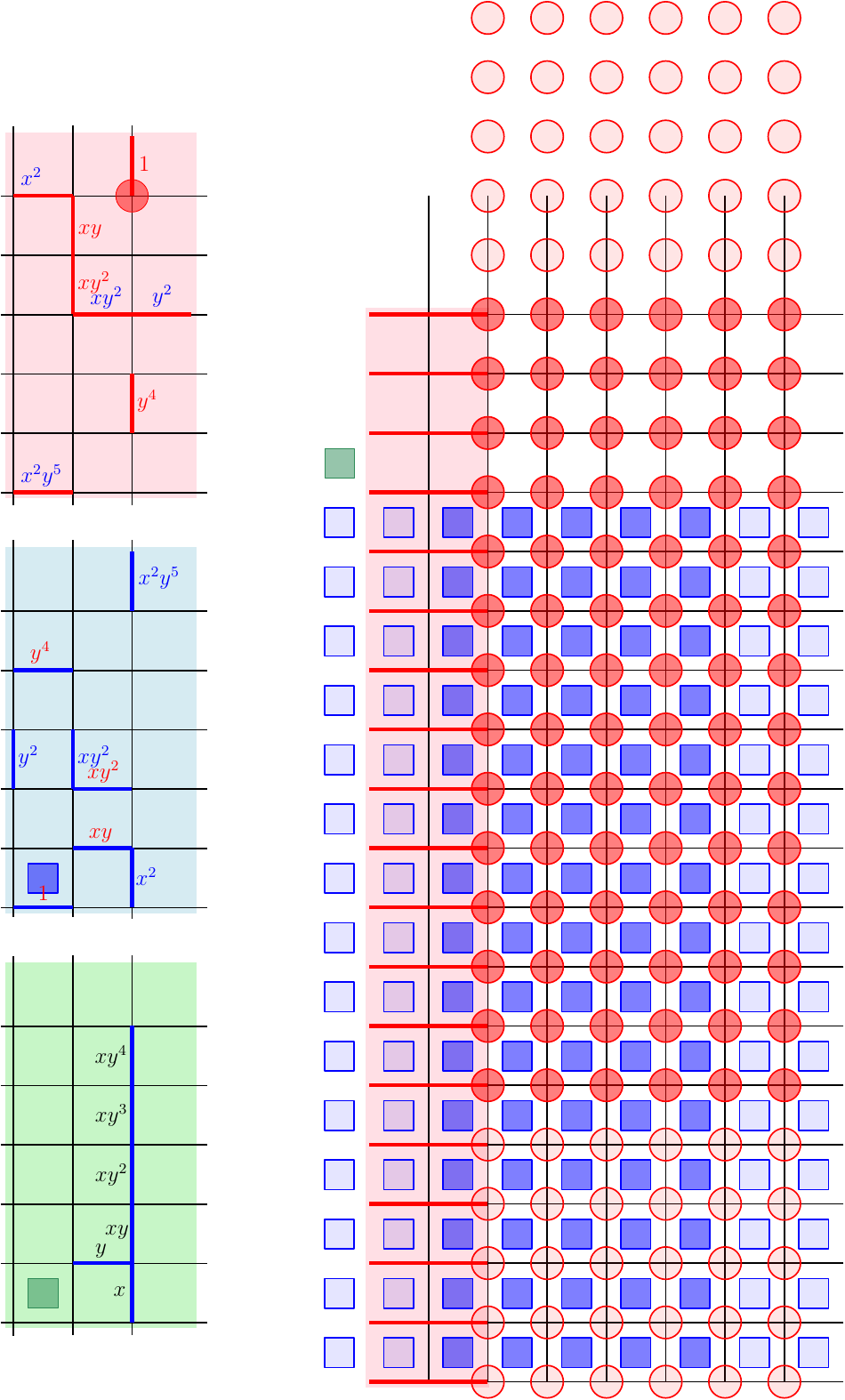}
  \caption{As in Fig.~\ref{fig:tile-corner} but for polynomials in
    Example \ref{ex:edge-off}. Dark and light blue squares (bulk \&
    edge $Z$ generators) and red circles (bulk \& edge $X$
    generators), and the additional $X$ generator of weight $38$ shown
    as the red ladder with pink shading are the stabilizer generators
    produced by Algorithm \ref{alg:guaranteed}. Sea-green tile is one
    of the 16 mutually-degenerate small-weight $Z$-logical operators
    in the code $[[292,13,14/6]]$ obtained by removing the high-weight
    $X$ generator. Promoting this logical operator to the stabilizer
    group gives a code with parameters $[[292,12,14/14]]$ identical to
    those of the original LEC code discussed in Sec IV.D of Ref.~\onlinecite{Liang-Eberhardt-Chen-2025}.}
  \label{fig:edge-off}
\end{figure}
\begin{example}
  \label{ex:edge-off} Consider a pair of binary weight-four
  polynomials 
  $a=y^2+xy^2+x^2(1+ y^5)$, $b=1+y^4+xy(1+y)$
  from Sec.~IV.D in Ref.~\onlinecite{Liang-Eberhardt-Chen-2025}, which
  gives, in particular, an LEC code with parameters $[[292,12,14]]$.
  These polynomials violate the vertical edge
  condition~(\ref{eq:edge-vert}), namely
  $\gcd\biglb(a_{\rm max}(y),b_{\rm max}(y)\bigrb)=1+y$, which is
  expected to give bounded-weight $Z$ chains on the left vertical edge.
  Algorithm~\ref{alg:guaranteed} with $L_x=8$, $L_y=19$ gives a code
  with parameters $[[292,12,14/16]]$, with $75$ bulk and $60$ boundary
  $Z$ stabilizer generators shown in Fig.~\ref{fig:edge-off} as dark
  and light blue squares, plus the total of $145$ $X$ generators: $84$
  bulk (dark red circles), $60$ boundary (light red circles), and an
  additional high-weight $X$ generator shown as a red ladder.
  Removing this generator gives a code $[[292,13,14/6]]$ with 16
  mutually-degenerate $Z$ codewords of weight $6$ shown in
  Fig.~\ref{fig:edge-off} as a sea-green tile.
  
  These codewords are exactly consistent with the analysis in
  Sec.~\ref{sec:strip-smooth-rough}. Namely, here the common factor
  is $1+y$, and we can take $a'_{\rm max}=1+y+y^2+y^3+y^4$,
  $b'_{\rm max}=y$ in a dual version of
  Eq.~(\ref{eq:banded-H1-prime}). Adding any of these boundary
  logical operators to the stabilizer group gives a code with
  parameters $[[292,12,14/16]]$ identical to the original LEC code.
  As an option, all $15$ $Z$ generators on the left edge may be
  replaced with these modified boundary operators (i.e., painting the
  leftmost column of light-blue squares in Fig.~\ref{fig:edge-off}
  with sea-green color.)
\end{example}
\begin{example}\label{ex:edge-off-dual}
  Consider a pair of binary weight-four
  polynomials 
  $a=x^2+yx^2+y^2(1+ x^5)$, $b=1+x^4+yx(1+x)$, i.e., the polynomials
  from Example \ref{ex:edge-off} with swapped variables, and $L_x=19$,
  $L_y=8$, equivalent to swapping rough and smooth boundaries. In this
  case $\gcd\biglb(a_{\rm max}(x),b_{\rm
    max}(x)\bigrb)=1+x$. Algorithm~\ref{alg:guaranteed} gives a code
  with parameters $[[314,13,6/12]]$, with 16 weight-$6$ logical $X$
  operators at the edge. Adding one of them to the stabilizer group
  gives a code with parameters $[[314,12,14/14]]$.  Exactly the same
  code is obtained by starting from a horizontal strip with modified
  edge operators, see Eq.~(\ref{eq:banded-H1-prime}).
\end{example}

To conclude this section, Algorithm \ref{alg:guaranteed} as stated, in
addition to the bulk condition in Lemma \ref{th:bulk}, requires the
additional edge conditions (\ref{eq:horizontal-edge-condition-inv}),
(\ref{eq:edge-vert}). For such polynomials, it is expected to give
codes similar to those from the LEC algorithm, namely, with the same
dimension, and similar distances that depend on the exact choice of
the corner generators in the latter algorithm.

The main advantage of Algorithm \ref{alg:guaranteed} is that it is
automatic, i.e., does not depend on an ad-hoc prescription for
choosing corner stabilizer generators, which simplifies a search of
interesting code families by exhaustive enumeration. In addition,
Algorithm \ref{alg:guaranteed} works with tilted boundaries.

On the other hand, since it relies on Gauss elimination to construct
plaquette generators, there is no guarantee that Algorithm
\ref{alg:guaranteed} gives a code with all stabilizer generators of
sufficiently small weight, or even that such a set of generators can
be constructed. In practice, the code candidates with interesting
parameters, e.g., returned by an exhaustive search with Algorithm
\ref{alg:guaranteed}, need to be examined more carefully to select
polynomial pairs which give code families with suitable sets of
stabilizer generators.

We note also that code length may be additionally reduced by lattice
grafting\cite{Liang-Eberhardt-Chen-2025}, possibly at the price of
increasing the generator weights. We do not discuss lattice grafting
in this work as it tends to produce irregular stabilizer generators
which makes circuit implementation difficult.

\subsection{Numerical results}
\label{sec:num}

We implemented Algorithm \ref{alg:guaranteed} in Python, utilizing the
{\tt vecdec} package for distance
calculations\cite{Pryadko-2025-vecdec}. We performed an extensive
search over inequivalent pairs of mutually prime bivariate polynomials
with weights $\wgt a = \wgt b=3$, where the absolute values of the $x$
and $y$ monomial degrees did not exceed $D_{\rm max}=3$. For each
polynomial pair, after verifying the corresponding edge conditions, we
constructed both regular (rectangular) and rotated BBS codes with
strip dimensions $L_1, L_2\le 30$, focusing on codes of lengths
$n\le n_{\rm max}=300$ of dimension $k\le k_{\rm max}=13$.  We
optimized the aspect ratios of these regions to balance the CSS
distances $d_X$ and $d_Z$. Provided the region size was sufficiently
large (avoiding finite-size effects where $k<\kappa$), all generated
code families strictly reproduced the expected topological dimension
from Eq.~(\ref{eq:horiz-vert}).

The resulting codes are cataloged in Appendix~\ref{app:num-table}. For
each $(k,d)$ pair, we report the minimum achieved block length
$n_{\rm min}(k,d)$, the generating polynomials, the boundary
orientation (regular or rotated), and the ratio $kd^2/n$ which
compares the encoding rate with that of rotated surface codes.  The
structure of generators of one of the constructed codes is illustrated in
Fig.~\ref{fig:factor-max}.

\begin{figure}[htbp]
  \centering
  \includegraphics[width=0.90\columnwidth]{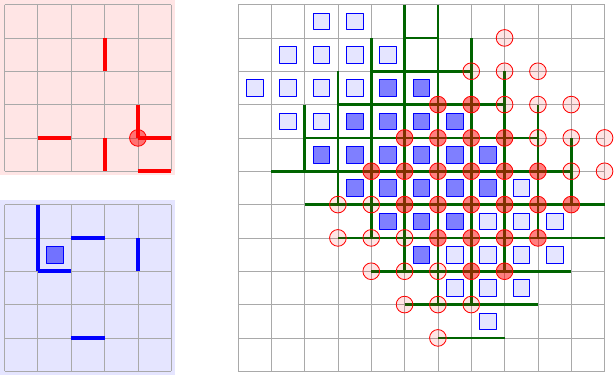}\hfill\strut
  \caption{As in Fig.~\ref{fig:surf-lattice} but for a BBS code
    $[[103,10,5]]$ with diagonal boundaries constructed from binary
    polynomials {\color{blue} $a=1+y^3+xy^2$} and
    {\color{red}$1+x^3+x^3y$}.  Here thick green lines show the edges
    corresponding to the qubits in the code.  Some other codes in this
    family have parameters $[[66,10,4]]$, $[[125,10,6]]$,
    $[[148, 10, 7]]$, $[[186, 10, 8]]$, $[[198,10,9]]$,
    $[[227,10,10]]$, $[[278,10,11]]$, and $[[292,10,12]]$ (in
    Tab.~\ref{tab:num-table}).  The last four codes in this list have
    $kd^2/n>4$.}
  \label{fig:factor-max}
\end{figure}
Our search space significantly expands upon the one explored in
Ref.~\onlinecite{Liang-Eberhardt-Chen-2025} by relaxing restrictions
on the polynomial form and explicitly incorporating diagonal
boundaries. Consequently, we found codes with shorter block lengths
for nearly all $(k,d)$ pairs, see Appendix \ref{app:num-table}.  More
remarkably, as highlighted in Table~\ref{tab:num-cmp}, in two cases,
our regular, ungrafted codes are shorter than the highly optimized,
lattice-grafted codes presented in
Ref.~\onlinecite{Liang-Eberhardt-Chen-2025}.

\begin{table}[htbp]
  \centering
  \begin{tabular}{c c | c | c | c}
    \hline\hline
    $k$ & $d$ & Orig.\ $n$ \cite{Liang-Eberhardt-Chen-2025} & Grafted $n$ \cite{Liang-Eberhardt-Chen-2025} & {This work $n$} \\ \hline
    6 & 4 & 54 & 44 & \textbf{42}\\
    6 & 6 & 88 & 78 & \textbf{72}\\
    7 & 7 & 131 & \textbf{107} & 126\\
    8 & 9 & 188 & \textbf{173} & 184\\
    8 & 12 & 288 & \textbf{268} & 288\\ \hline\hline
  \end{tabular}
  \caption{Brief comparison of minimum block lengths $n$ for selected
    BBS codes found using Algorithm \ref{alg:guaranteed} versus the
    original and lattice-grafted codes from
    Ref.~\onlinecite{Liang-Eberhardt-Chen-2025}.}
  \label{tab:num-cmp}
\end{table}

\section{Conclusion}
\label{sec:conclusion}

To summarize, we have conducted a careful analysis of the structure of
codewords for families of bivariate-bicycle (BB) codes on the infinite
plane and in finite-width strips with different boundary conditions,
with a focus on sublattice translation and group isomorphism
symmetries. We identified the structure of balanced boundaries
optimal for specific polynomials, and exact edge conditions on the
polynomials which guarantee that stabilizer generators at the boundary
can be obtained merely by truncation of the bulk generators. When the
edge conditions are not satisfied, we have also constructed the
polynomials generating the required boundary generators.

One unexpected result is that key features of this construction are
related to the number and location of common roots of the polynomials
defining the code. Namely, the common dimension of a family of BBS
codes based on a pair of mutually prime polynomials
$a,b\in F[x^{\pm1},y^{\pm1}]$ (equal to the total quantum dimension
associated with the ground-state TO on the infinite plane) is
associated with the number of finite and non-zero common roots of the
two polynomials in the extension field, counted with algebraic
multiplicity. This number remains invariant under symmetries of the
parent chain complex on $\mathbb{Z}^2$: CSS duality, sublattice
translations, or arbitrary invertible variable substitutions in the
polynomials corresponding to automorphisms of the infinite plane.

On the other hand, such roots located at $x=0$ or infinity, and $y=0$
or infinity, respectively, signify that specially modified stabilizer
generators on the left, right, bottom, or top boundaries are required.
Unlike the finite non-zero roots, these boundary roots are only
invariant to sublattice translations but not to more general
invertible affine transformations of $\mathbb{Z}^2$. Such
transformations map between pairs of parallel boundaries running in
different directions and, in particular, between pairs of tilted and
horizontal boundaries. This gives a general approach to constructing
edge generators for a boundary along an arbitrary lattice direction,
and, in particular, predicting when simple truncation would be
sufficient.

Based on these results, we formulate an algorithm for constructing
surface versions of bivariate-bicycle codes, including codes with
rectangular, diagonal, and, more generally, arbitrarily tilted
boundaries. A notable advantage of the algorithm is that it preserves
the sublattice translation symmetry characteristic of the BB codes:
polynomials rescaled by monomial factors always produce equivalent
codes. In addition, the algorithm comes with a guarantee: any pair of
polynomials that satisfies explicitly stated bulk and edge conditions
gives a family of codes with fixed $k$ and a distance increasing
without bound, asymptotically as $\mathcal{O}(n^{1/2})$.

This construction generalizes the results on tile codes from
Refs.~\cite{Steffan-etal-Eberhardt-2025,Breuckmann-etal-Steffan-2025},
in the sense that all tile codes can be exactly reproduced by
Algorithm \ref{alg:guaranteed}, but not all codes generated by our 
algorithm can be obtained as tile codes. On the other hand, our 
algorithm does not inherently minimize overhead to the extent of the 
lattice grafting techniques in Ref.~\onlinecite{Liang-Eberhardt-Chen-2025}. 
While two of our regular codes outperform the grafted codes constructed in 
Ref.~\onlinecite{Liang-Eberhardt-Chen-2025} (see Sec.~\ref{sec:num}), 
this success is primarily due to the much broader search space enabled 
by Algorithm \ref{alg:guaranteed}, particularly the inclusion of rotated 
boundaries and unrestricted polynomials. To achieve the absolute minimum 
block length, a combined approach should be used: heavy optimization 
protocols, such as lattice grafting, should be applied to a select few 
optimal polynomial pairs identified via Algorithm \ref{alg:guaranteed}. 
In particular, all tile codes and many BBS codes from Algorithm 
\ref{alg:guaranteed} have edge stabilizer generators which can be 
obtained simply by trimming those in the bulk. This allows for 
time-optimal measurement schemes similar to those for surface codes, 
where edge and bulk generators share the addressing patterns of the 
corresponding BB codes.

\begin{acknowledgments}
  This work was supported in part by the NSF award 2112848 (LPP).
\end{acknowledgments}

\appendix

\section{Dimension of the quotient space and common roots of polynomials}
\label{app:root-count}

In general, the relation between the dimension of a finite-dimensional quotient
space in a ring of multivariate polynomials with coefficients in a
field $F$ and the number of common roots of a polynomial ideal can be
stated as follows [See, e.g., Ref.~\onlinecite{Cox-Little-OShea-2005}]:
\begin{theorem}
  Let $\cal J$ be a zero-dimensional ideal in the polynomial ring
  $F[x_1, \ldots, x_n]$, and let $\overline{F}$ be the algebraic
  closure of $F$. The quotient ring $F[x_1, \ldots, x_n]/{\cal J}$ is a
  finite-dimensional vector space over $F$, and its dimension is
  exactly equal to the total number of common roots of the ideal in
  $\overline{F}^n$, where each root is counted with its algebraic
  multiplicity.
\end{theorem}

In particular, this applies to any abelian two-block group-algebra
code\cite{Kalachev-Panteleev-2020,Wang-Lin-Pryadko-2023,Lin-Pryadko-2023},
a special case of lifted-product codes over an abelian group
${\cal G}$. For a finite group presented with $n$ commuting generators
and $m$ relators, using $\bs x\equiv (x_1,x_2,\ldots, x_n)$,
$$
{\cal G}=\left\langle \bs x\mid r_1(\bs x)=r_2(\bs x)=\cdots =r_m(\bs x)=1\right\rangle_{\rm abelian},
$$
and a finite field $F$, such a code can be defined by any pair of
polynomials $a,b\in F[\bs x]$. Namely, the ideal in question is
${\cal J}=\langle a,b,r_1-1,\ldots, r_m-1\rangle$, while the dimension
of the code is given by\cite{Liang-Liu-Song-Chen-2025}:
$$
k_{\text{two-block abelian}}=2\dim_F\frac{F[\bs x]}{ {\cal J}}.
$$

For BB codes\cite{Bravyi-etal-Yoder-2023} defined by a
pair of polynomials $a,b\in F[x,y]$ on an $L_x\times L_y$ torus with
periodicity vectors along the axes, this reduces to, 
$$
k_{\text{BB}}(a,b;L_x,L_y)=2\dim_F\frac{F[x,y]}{ \langle
  a,b,x^{L_x}-1,y^{L_y}-1\rangle}.
$$
where the factor of $2$ is associated with the genus of the torus.
Related but different is the total quantum dimension $k_{\rm TO}$
associated with the TO in the ground state manifold of the
corresponding Hamiltonians on the infinite plane, which is defined in
terms of topological excitations.  Specifically for the BB codes on
the infinite plane\cite{Liang-Liu-Song-Chen-2025}, working in the ring
of Laurent polynomials, where zero coordinates are strictly forbidden,
one has\cite{Liang-Liu-Song-Chen-2025}:
$$
k_{\text{TO}} = \dim_F{F[x^{\pm 1},y^{\pm 1}]\over
  \langle a,b\rangle}.
$$
This can be related to the ground-state dimension
$k_{\text{BB}}(a,b;L_x,L_y)$ by the supremum:
$$
k_{\text{TO}} = \frac{1}{ 2}k_{\text{BB}}(a,b;\mathbb{Z}^2)={1\over2} \sup_{L_x,L_y\in \mathbb{N}}
k_{\text{BB}}(a,b;L_x, L_y).
$$
Here any roots with $x=0$ or $y=0$ are automatically excluded for any
$L_x$, $L_y$. 

More directly, this dimension can be computed using Rabinowitsch's
trick by introducing an auxiliary variable $t$ and an extra equation, $txy=1$, to
enforce non-zero coordinates:
$$
k_{\text{TO}}(a,b;\mathbb{Z}^2)= \dim_F{F[x,y,t]\over
  \langle a,b,xyt-1\rangle}.
$$
Specifically, the dimension on the right-hand side can be computed by constructing a
Gr\"obner basis of the ideal $ \langle a,b,xyt-1\rangle$ and counting
the standard monomials with respect to that basis.

\section{Proof of Statement \ref{th:SNF-Mx}}
\label{app:proof-th:SNF-Mx}

The original matrix $ {\sf M}_X$ with elements
in $F[x]$ has $N\equiv m+\Delta$ rows and $2m$ columns. By definition
of the Smith normal form, the $k$-th invariant factor $d_k$ of ${\sf M}_X$ is
$$d_k = \frac{\Delta_k}{\Delta_{k-1}},$$ where $\Delta_k$ is
the greatest common divisor (GCD) of all $k \times k$ minors of
${\sf M}_X$, and $\Delta_0 \equiv 1$. Thus, to prove that ${\sf M}_X$
has $N$ non-zero SNF invariants, we just need to find a non-zero minor
of size $N\times N$, either in the matrix ${\sf M}_X$, or in any
matrix obtained from ${\sf M}_X$ by invertible row and/or column
transformations.

Operating in $F[x]$, a GCD domain, denote
$$h_0(x)\equiv \gcd(a_0,b_0), \quad
\dot a_0\equiv a_0/h_0,\quad 
\dot b_0\equiv b_0/h_0,$$ and let $r(x)$, $s(x)$ be the corresponding
B\'ezout coefficients such that $r(x) a_0(x)+s(x)
b_0(x)=h_0(x)$. Multiplying $ {\sf M}_X$ by the block-diagonal
unimodular (and thus invertible over $F[x]$) matrix
$$
U=\diag(Q,Q,\dots,Q),\quad Q\equiv \left(
  \begin{array}[c]{cc}
    r& -\dot b_0\\ s&\dot a_0
  \end{array}
\right)
$$
from the right gives a matrix $M_1\equiv M_1^{(N,2m)}$ with a
symmetric block structure similar to that of ${\sf M}_X$, but with the mapping
$a_0\mapsto h_0$ and $b_0\mapsto 0$ in the first non-zero row of each
pair of columns. To form a non-zero minor, select the first
$m-\Delta$ odd columns and the last $2\Delta$ columns of $M_1$, which
gives a square matrix $M_{1,\rm sub}^{(N,N)}$ featuring a lower-triangular
submatrix in the first $m-\Delta$ rows. The determinant is readily
simplified by repeated first-row expansion, yielding:
$$
\det M_{1,\rm sub}^{(N,N)}=h_0^{m-\Delta}\det {\sf M}_X^{(\Delta)}.
$$
The first term in the product on the right-hand side is non-zero by
the assumption that both $a_0\neq 0$ and $b_0\neq0$, while the
remaining determinant is also non-zero: according to
Eq.~(\ref{eq:det-Mx}), this is guaranteed by the assumption that $a$
and $b$ are mutually prime.

If we now enforce the first additional condition in Eq.~(\ref{eq:SNF-Mx}),
equivalent to $h_0(x)=1$, we return to the matrix $M_1^{(N,2m)}$, which
now contains $(1,0)$ at the top of each two-column block. We also
construct the transformed matrix $M_2\equiv U^{-1}{\sf M}_Z^T$, which has $(0,1)^T$
at the top of each column. By construction, orthogonality is preserved: $M_1M_2=0$.

Next, we use row transformations to construct $M_2'$, a reduced row
echelon form of $M_2$, and apply the corresponding orthogonal column transformations
$M_1\mapsto M_1'$ to render the first $m-\Delta$ even-numbered columns
of $M_1'$ to zero while preserving the product $M_1'M_2'=0$. The
third step uses row transformations to construct $M_1''$, a
reduced row echelon form of $M_1'$ within the first $m-\Delta$ odd columns.

Crucially, because of the strict block-diagonal elimination, these transformations leave unchanged the
lower-right square portion of $M_1^{(N,2m)}$ of size $2\Delta$, which
is row-/column-transformation similar to the original square boundary matrix
${\sf M}_X^{(\Delta)}$. The corresponding rows of $M_2'$ are
identically zero by construction. After suitable row and column
transformations, the $2\Delta\times 2\Delta$ submatrix in the
lower-right corner of $M_1''$ is replaced by its SNF form, with the
diagonal formed strictly by the SNF invariants of ${\sf M}_X^{(\Delta)}$.

Up to a column permutation, this completes a transformation of the
original matrix ${\sf M}_X$ into the claimed diagonal form, featuring
exactly $m-\Delta$ unit invariants appended to the original $2\Delta$
non-zero invariants of ${\sf M}_X^{(\Delta)}$.  The proof for the
other edge condition, $\gcd(a_\Delta,b_\Delta)=1$, follows identically
by reversing the row and column order. \hfill\qed

\section{Proof of Statement \ref{th:SNF-Mz}}
\label{app:proof-SNF-Mz}

Denote $N\equiv m-\Delta$ as the number of rows of the matrix
$M\equiv {\sf M}_Z$ over $F[x]$ [Note that Eq.~(\ref{eq:banded-H2})
gives the transposed matrix].  Similar to the previous proof, we just
need to look at the minors of size $N\times N$ spanning all rows of
$M$.

The first statement can be trivially obtained by selecting the minor
which contains the leftmost non-zero element from each row of $M$. As
the determinant of an upper-triangular matrix with non-zero elements
along the diagonal (of size $N\times N$), it is non-zero. This proves
$\Delta_N\neq0$, yielding exactly $N$ non-zero SNF invariants.

To show that all invariant factors $d_1, d_2, \dots, d_N$ are actually
equal to $1$, it is sufficient to prove that the GCD of all
$N\times N$ minors is $\Delta_N = 1$. We have already established that
$\Delta_N\in F[x]$ is non-zero. Assuming $\Delta_N$ is not a unit,
there must exist some non-constant irreducible polynomial
$q(x) \in F[x]$ that divides every single $N \times N$ minor of the
original matrix. Let $K$ be the quotient field
$F[x] / \langle q(x) \rangle$.

Denote by $\overline{M}$ the matrix $M$ with all its entries evaluated
modulo $q(x)$. Since $q(x)$ divides every $N \times N$ minor of $M$,
every $N \times N$ minor of $\overline{M}$ evaluates to $0$ in the
field $K$. This implies that the rows of $\overline{M}$ are linearly
dependent over $K$.

However, because the original bivariate polynomials $a$ and $b$ are
mutually prime, there is no polynomial $q(x)$ that can simultaneously
divide all $y$-coefficients $a_i(x)$ and $b_i(x)$, $0\le i\le
\Delta$. Therefore, for at least one index $i$, either
$a_i(x) \not\equiv 0$ or $b_i(x) \not\equiv 0 \pmod{q(x)}$. By
analyzing the banded, translationally invariant structure of $M$ [see
Eq.~(\ref{eq:banded-H2})], the shifted placement of these non-zero
coefficients guarantees that the rows of $\overline{M}$ remain
linearly independent over $K$. This contradicts the assumption that
$\overline{M}$ has deficient rank. Necessarily, $\Delta_N=1$,
completing the proof.\hfill$\qed$

\section{Proof of Statement~\ref{th:mix-edge-strip}}
\label{app:th:mix-edge-strip}

The proof of the first part  and the ``if'' part of the second are
similar to that in Appendix \ref{app:proof-th:SNF-Mx}, except here the
transformed matrices $M_1\equiv {\sf M}_XU$ and
$M_2\equiv U^{-1}{\sf M}_Z^T$ have enough unit pivot elements to
simultaneously transform them to block-diagonal matrices,
$M_1\mapsto (I_m,0)$ and $M_2\mapsto(0,I_m)^T$, where $I_m$ is an
$m\times m$ identity matrix.

To prove the ``only if'' part, notice that $\gcd(a_0,b_0)$ is a common
factor of the first row of ${\sf M}_X$. Since this matrix has $m$
rows, any minor of size $m$ must contain $\gcd(a_0,b_0)$ as a factor.
The proof goes similarly for ${\sf M}_Z^T$ which has $m$ columns and has
$\gcd(a_0,b_0)$ as a common factor of elements in the last column.
\hfill\qed

\section{Number of localized solutions at a lower smooth edge}
\label{app:root-edge-count-proof}

Here we show that the number of $X$ co-chain solutions localized at a
lower horizontal smooth edge (in a sufficiently wide strip) coincides
exactly with the number of common roots of polynomials $a(x,y)$,
$b(x,y)$ at $y=0$, counted with algebraic multiplicity.  The number of
localized co-chains is the degree spread of the product of SNF
invariants of the matrix (\ref{eq:banded-mix-M1}).  We are going to
modify this matrix row-by-row, factoring out common factors in each
row, which gives a matrix with the top $m$ rows identical to the top
row block of ${\sf M}_X'$ in Eq.~(\ref{eq:banded-H1-prime}).

Using arguments similar to those in Appendix \ref{app:proof-SNF-Mz},
one can see that the transformed matrix has exactly $m$ unit SNF
invariants; the number of localized solutions is given by the degree
of the product of the extracted factors,
$$
f(x)=\prod_{i=0}^{m-1} h_i(x).
$$
As in Eq.~(\ref{eq:h0}), $h_{j-1}(x)$, $j>0$, is the GCD of the
elements in the $j$-th row of matrix (\ref{eq:banded-mix-M1}), or,
equivalently, the common GCD of the first $j$ expansion terms of the
two polynomials in powers of $y$; we assume $h_i(x)$ to be monic with
non-zero free terms.  Denoting $k_i\equiv \deg h_i(x)$, we see that
along $y=0$, the polynomials have $k_0-k_1$ common roots of
multiplicity one in $y$, $k_1-k_2$ common roots of multiplicity two,
etc, which gives for the number of solutions localized at the edge
\begin{eqnarray*}
  \deg f(x)& =& k_0+k_1+k_2+\ldots +k_{m-1}\\
           &=& mk_m + \sum_{j=1}^{m}j (k_{j-1}-k_j).   
\end{eqnarray*}
The sum in the second term explicitly counts each root weighted by its exact
intersection multiplicity.  By assumption, the polynomials have no
common factors, thus $k_{i}=0$ for any $i>\Delta$, the maximum $y$
degree of the polynomials.  Thus, for any $m>\Delta$, all roots at
$y=0$ with their multiplicities are included in the count.\hfill$\qed$

\section{Proof of Statement \ref{th:modified-exactness}}
\label{app:proof-th:modified-exactness}

The proof follows arguments similar to those in Appendix
\ref{app:root-edge-count-proof}, except here we must extract common
factors at both boundaries of the strip, corresponding to common roots at
$y=0$ and $y=\infty$.

The original strip with bare edges supports $\kappa_{\rm horiz}$
solutions [see Eq.~(\ref{eq:kappa-horiz})], given by the degree spread
of the $y$-resultant of the two polynomials. This counts all roots
$(x,y)$ such that $x \neq 0$, including those where $y$ is zero, infinite,
or finite. By modifying the bottom and top edges, we explicitly remove
the roots at $y=0$ and $y=\infty$, respectively, leaving only the roots with
both coordinates finite and non-zero.

\section{Univariate polynomials}
\label{app:univariate}
Consider the matrices (\ref{eq:banded-H1}) and (\ref{eq:banded-H2}) with
$m=\Delta$ in the special case where one of the polynomials depends
only on one variable---say, $b(x,y)=b_0(x)$---while the other polynomial
has degree $\Delta\equiv \deg_y a>0$. As a reminder, here we assume
polynomials $a,b\in F[x,y]$ non-vanishing identically at $y=0$, $a_0(x)\neq0$ and
$b_0(x)\neq0$.

With $\deg_y b(x)=0$, the determinant of ${\sf M}_X$ with $m=\Delta$
equals $[a_\Delta(x)]^\Delta [b_0(x)]^\Delta$. Assuming a unit
leading coefficient $a_\Delta(x)=1$ as required by the condition
$\gcd(a_\Delta(x),b_\Delta(x))=1$, only the second term contributes to the
degree spread, yielding:
\begin{equation}
  \kappa_1=\deg b_0(x)\,\Delta\equiv\deg_x b\,\deg_y a, \quad \text{for } \deg_y b=0.
  \label{eq:kappa-univariate}
\end{equation}

Alternatively, consider the horizontal edge conditions in invariant
form (\ref{eq:horizontal-edge-condition-inv}), explicitly
$\gcd(a_0,b_0)=\gcd(a_\Delta,b_0)=1$. A minimum-width strip contains
$\Delta+1$ horizontal edges and one vertical edge. The matrices take the form:
$$
    {\sf M}_X
    =\left(
        \begin{array}[c]{c|ccccccc}
          a_0&b_0&   \\ 
          a_1&  &b_0 \\
          {\scriptscriptstyle\boldsymbol\vdots}
              &  & &{\scriptscriptstyle\boldsymbol\ddots} \\ 
          a_{\Delta}&  & & & b_0  
        \end{array}\right),\quad
    {\sf M}_Z^T
    = \left(
        \begin{array}[c]{c}
          b_0  \\ \hline 
          - a_0 \\
          - a_1\\ 
          {\scriptscriptstyle\boldsymbol\vdots}\\
          - a_{\Delta}      \end{array}\right). 
$$
The bulk condition implies that the matrix ${\sf M}_Z$ has only one
unit SNF invariant, while the product of SNF invariants of ${\sf M}_X$
is given by Eq.~(\ref{eq:kappa-univariate}).

Finally, consider the case $\deg_y a =\deg_y b=0$. The parent complex
on the infinite plane separates into independent horizontal strips,
each including a row of horizontal and a row of vertical edges. Under
the horizontal edge condition (\ref{eq:edge-condition-m1}), with $m=1$,
both ${\sf M}_X$ and ${\sf M}_Z$ have exactly one unit invariant, yielding
$\kappa_1=0$. Otherwise, if 
\begin{equation}
\gcd \biglb(a_0(x),b_0(x)\bigrb)=h_0(x)\neq 1, \label{eq:gcd-uni}
\end{equation}
the only non-zero SNF invariant of each matrix is $h_0(x)$. This gives
$\deg h_0(x)$ infinite horizontal chains and the same number of
finite-weight co-chains, in agreement with Lemma~\ref{th:bulk}. That is,
a BBS code family constructed from a fixed pair of univariate
polynomials $a,b\in F[x]$ is either trivial or has a distance bounded
by a constant.

\section{Proof of Statement \ref{th:code-rectang}}
\label{app:th:code-rectang}

With the results on the structure of logical operators established in
Sections \ref{sec:covering-plane} and \ref{sec:with-boundaries}, this
statement follows from $Z$-shortening and $Z$-puncturing bounds
(Lemmas 1 and 2 in Ref.~\onlinecite{Zeng-Pryadko-hprod-2020}), which
must be adapted to an original code defined on an infinite strip.

The most straightforward resolution is to consider an original code on a horizontal
strip with periodic boundary conditions along the $x$-direction, with
period $L$ chosen so that all non-zero roots of the resultant in
Eq.~(\ref{eq:kappa-horiz}) are among the roots of $x^L-1$. Given
$r(x)\equiv\res_y\biglb(a(x,y),b(x,y)\bigrb)\in F[x]$, we define
$f(x)=r(x) x^{-\deg_{\rm min} r(x)}$ so that $f(0)\neq 0$.
The minimum $L$ such that $x^L-1$ is divisible by $f(x)$ is given
by $L=\ord f(x)$, the order of the polynomial $f(x)$. With
$F\equiv \text{GF}(q)$ a Galois field of size $q=p^m$, $\ord f$ is found by
factorizing $f(x)\in F[x]$ into irreducible polynomials:
$f(x)=c\,\prod_j [f_j(x)]^{e_j}$. The order is then given by
$\ord f =\lcm\biglb(\ord f_1^{e_1}, \ord f_2^{e_2}, \ldots\bigrb)$,
with $\ord f_i^{e_i}=p^{t_i}m_i$, where $t_i=\lceil\log_p e_i\rceil$ and $m_i\equiv \ord f_i(x)$ is the 
minimum divisor of $q^{d_j}-1$ such that $x^{m_i}-1$ is divisible by $f_i(x)$.

With this period $L$ along the $x$-direction, the number of independent
solutions in the infinite strip (with balanced smooth boundaries and 
sufficiently large $L_y$) is identical to the dimension of a
quasi-cyclic CSS code constructed from block matrices similar to
those in Eqs.~(\ref{eq:banded-mix-M1}) and (\ref{eq:banded-mix-M2}),
where each element is replaced by an $L\times L$ circulant matrix:
$$
k=\deg \gcd\left(x^L-1,\res_y(a,b) \right) = \deg \res_y(a,b).
$$
Furthermore, the horizontal edge conditions guarantee that any non-trivial
horizontal chain extends fully across the width, while any
non-trivial vertical co-chain connects the two vertical boundaries;
the full set of such co-chain solutions can be found in a vertical
strip of horizontal size greater than or equal to $\Delta_x$.

Now, let $H_X$ and $H_Z$ be the CSS stabilizer generator matrices of the
constructed quasi-cyclic code $\mathcal{Q}_L$, and let $L_X$ and $L_Z$ be
the corresponding logical generator matrices [see Eq.~(\ref{eq:logical})]. 
We form two pairs of mutually dual matrices,
$$
\begin{pmatrix}H_X \\ L_X\end{pmatrix},
H_Z\quad\text{and} \quad H_X, \begin{pmatrix}H_Z \\ L_Z\end{pmatrix}.
$$
Matrices in each pair remain dual after the first matrix is punctured
and the second matrix is shortened to the strip of size $\Delta_x < L_x < L-\Delta_x$, 
which entirely supports the chosen $X$ co-chain representatives in the rows of $L_X$. 
All non-zero elements of the matrix $L_X L_Z^T$ are preserved after puncturing, 
guaranteeing that the code dimension remains invariant.

The lower bounds on the CSS distances follow from the fact that any
non-trivial chain must connect distinct rough edges, and any co-chain 
must connect smooth edges, with at most $\Delta_x$ horizontal and $\Delta_y$
vertical intervals between subsequent non-zero bits.

\begin{widetext}
\section{Summary of numerical results}
\label{app:num-table}

\begin{longtblr}[ caption = {Parameters of BBS codes with $n\le 300$,
    $k\le 13$}, label = {tab:num-table}, ]{ colspec = {|c|c|c|c|c|c|c|},
    width = 0.75\linewidth, rowhead = 1, row{odd} = {gray9}, row{even}
    = {brown9}, row{1} = {olive9}, hlines, }
    $k$&$d$&$n_{\rm min}$&$a(x,y)$&$b(x,y)$&rotated&$kd^2/n$\\ 
    6& 4& 42&$1+y^2+xy$&$1+x^2+x^3$& Yes&2.286\\
    6& 5& 72&$1+xy^{-3}+xy^{-1}$&$1+y+xy^3$&No&2.083\\
    6& 6& 72&$1+xy^{-1}+x^2y$&$1+xy^{-2}+x^2y^{-1}$&
    No&3.000\\
    6& 7& 112&$1+xy^{-1}+x^2y$&$1+xy^{-2}+x^2y^{-1}$&
    No&2.625\\
    6& 8&
  128&$1+xy^{-1}+x^2y$&$1+xy^{-2}+x^2y^{-1}$&
  No&3.000\\
  6& 9&
  160&$1+xy^{-1}+x^2y$&$1+xy^{-2}+x^2y^{-1}$&
  No&3.038
  \\
  6& 10& 180&$1+xy^{-1}+x^2y$&$1+xy^{-2}+x^2y^{-1}$& No&3.333\\
  6& 11& 216&$1+xy^{-3}+xy^{-2}$&$1+xy^{-1}+x^2$& Yes&3.361\\
6&	12&
247&$1+x^2y^{-2}+x^2y^{-1}$&$1+xy+x^2y$&No&3.498\\
  6& 13& 275&$1+xy^2+x^2y^{-1}$&$1+y^2+xy$& No&3.687\\
  7& 4& 52&$1+xy^{-3}+xy^{-2}$&$1+x^2y^{-1}+x^3y^{-2}$& Yes&2.154\\
  7& 5& 81&$1+x^2y^{-1}+x^2y$&$1+x^{-1}y^2+xy^2$&Yes&2.160\\
  7& 6& 96&$1+x^2y^{-2}+x^2y^{-1}$&$1+xy^2+x^2y$&  No&2.625\\
  7& 7& 126&$1+x^2y^{-2}+x^2y^{-1}$&$1+xy^2+x^2y$&  No&2.722\\
  7& 8&
  156&$1+x^2y^{-1}+x^2y$&$1+xy^2+x^2$&  Yes&2.872\\
  7& 9&
  166&$1+x+x^2y$&$1+y^2+xy^{-1}$& No&3.416\\
  7& 10& 201&$1+xy^{-2}+x^2y$&$1+x+x^2y^{-1}$&Yes&3.483\\
  7& 11& 229&$1+xy^{-1}+xy^2$&$1+y+x^2y$&Yes&3.699\\
  7& 12& 257&$1+xy^{-1}+xy^2$&$1+xy^{-2}+x^2y^{-1}$&No&3.922\\
  7& 13& 294&$1+xy^{-1}+xy^2$&$1+y+x^2y$&Yes&4.024\\
  8& 4& 50&$1+xy^{-3}+xy^{-2}$&$1+x^2y^{-1}+x^3y^{-1}$&Yes&2.560\\
  8& 5& 72&$1+xy^{-3}+xy^{-2}$&$1+x^2y^{-1}+x^3y^{-1}$&Yes&2.778\\
  8& 6& 78&$1+xy^{-3}+xy^{-2}$&$1+x^2y^{-1}+x^3y^{-1}$&Yes&3.692\\
  8& 7& 137&$1+x+x^2y^2$&$1+y^2+xy^{-1}$&No&2.861\\
  8& 8& 144&$1+xy^{-3}+xy^{-2}$&$1+x^2y^{-1}+x^3y^{-1}$&Yes&3.556\\
  8& 9& 184&$1+xy^{-1}+x^2y^2$&$1+xy^2+x^2$&No&3.522\\
  8& 10& 224&$1+xy^2+x^2y^{-1}$&$1+y+x^2y$&Yes&3.571\\
  8& 11& 271&$1+xy^2+x^2y^{-1}$&$1+y+x^2y$&Yes&3.572\\
  8& 12& 288&$1+x^2y^{-2}+x^2y^{-1}$&$1+x+x^2y^2$&No&4.000\\
  9& 4& 56&$1+y^2+y^3$&$1+x^2+x^3$&Yes&2.571\\
  9&
  5& 98&$1+xy^{-3}+xy^{-2}$&$1+x^2y^{-1}+x^3$&Yes&2.296\\
  9&
  6& 110&$1+y^2+xy^{-1}$&$1+xy+x^2y^3$&No&2.945\\
  9&
  7& 142&$1+xy^{-3}+xy^{-2}$&$1+x^2y^{-1}+x^3$&Yes&3.106\\
  9&
  8& 178&$1+xy^{-3}+xy^{-2}$&$1+x^2y^{-1}+x^3$&Yes&3.236\\
  9&
  9& 206&$1+y+xy^3$&$1+xy^{-1}+x^3$&Yes&3.539\\
  9&
  10& 236&$1+xy^{-3}+xy^{-2}$&$1+x^2y^{-1}+x^3$&Yes&3.814\\
  9& 11& 250&$1+xy^{-3}+xy^{-2}$&$1+x^2y^{-1}+x^3$&Yes&4.356\\
  9& 12& 283&$1+xy^{-3}+xy^{-2}$&$1+x^2y^{-1}+x^3$&Yes&4.580\\
  10& 4& 65&$1+y+x^2y^{-1}$&$1+xy+x^3y^2$&No&2.462\\
  10& 5& 102&$1+y+x^2y^{-1}$&$1+x^2y+x^3y^2$&No&2.451\\
  10& 6& 114&$1+y+x^2y^{-1}$&$1+xy+x^3y^2$&No&3.158\\
  10& 7& 147&$1+y+x^2y^{-1}$&$1+xy+x^3y^2$&No&3.333\\
  10& 8& 184&$1+y+x^2y^{-1}$&$1+xy+x^3y^2$&No&3.478\\
  10& 9& 198&$1+y^3+x$&$1+x^2y^{-1}+x^3$&Yes&4.091\\
  10& 10& 227&$1+y^3+xy^2$&$1+x^3+x^3y$&Yes&4.405\\
  10& 11& 277&$1+y^3+xy^3$&$1+x^2y+x^3$&Yes&4.368\\
  10& 12& 292&$1+y^3+xy^2$&$1+x^3+x^3y$&Yes&4.932\\
  11& 4& 75&$1+y^2+xy^{-2}$&$1+x+x^3y^{-1}$&Yes&2.347\\
  11& 5& 115&$1+y^2+xy^{-2}$&$1+x^2+x^3y^{-1}$&Yes&2.391\\
  11& 6& 127&$1+y^2+xy^{-2}$&$1+xy^2+x^2y^3$&No&3.118\\
  11& 7& 176&$1+y^2+xy^{-2}$&$1+xy^2+x^2y^3$&No&3.062\\
  11& 8& 210&$1+y^3+x$&$1+x^2y^{-2}+x^3$&Yes&3.352\\
  11& 9& 234&$1+y^2+x^3y^2$&$1+x+xy^3$&Yes&3.808\\
  11& 10& 285&$1+y^3+x$&$1+x^2y^{-2}+x^3$&Yes&3.860\\
  12& 4& 77&$1+y^2+xy^{-2}$&$1+x+x^3$&Yes&2.494
  \\
  12& 5& 120&$1+y+x^3y^{-1}$&$1+xy+x^3y^2$&No&2.500\\
  12& 6& 144&$1+y^2+x^3y$&$1+xy^{-3}+x^2$&Yes&3.000  \\
  12& 7& 182&$1+y+x^3y^{-1}$&$1+xy+x^3y^2$&No&3.231\\
  12&  8& 216&$1+y+xy^{-3}$&$1+xy+x^3$&Yes&3.556\\
  12& 9& 260&$1+y^2+x^2y^{-1}$&$1+xy+x^2y^3$&No&3.738
  \\
  12& 10& 288&$1+y+x^3y^{-1}$&$1+xy+x^3y^2$&No&4.167\\
  13& 4& 85&$1+y^2+x^2y^{-1}$&$1+x+x^3y^2$&No&2.447\\
  13& 5& 126&$1+y^2+x^3$&$1+x^2y^{-3}+x^2$&Yes&2.579\\
  13& 6& 163&$1+y+x^3$&$1+xy^{-2}+xy^2$&Yes&2.871\\
  13& 7& 189&$1+y+x^3$&$1+xy^{-2}+xy^2$&Yes&3.370\\
  13& 8& 218&$1+y+x^3$&$1+xy^{-2}+xy^2$&Yes&3.817\\
  13& 9& 248&$1+y+x^3$&$1+xy^{-2}+xy^2$&Yes&4.246\\
  13& 10& 281&$1+y+x^3$&$1+xy^{-2}+xy^2$&Yes&4.626\\
\end{longtblr}
\end{widetext}
\bibliography{qc_all,lpp,more_qc,linalg,ldpc,spin,misc,citation,poly}

\end{document}